\documentclass{article}
\usepackage[a4paper,margin=1in]{geometry}
\usepackage[english]{babel}
\usepackage{xr}
\usepackage{float}

\usepackage{setspace}
\usepackage{tikz}
\usetikzlibrary{arrows,automata,positioning}
\usepackage{pgfplots}
\pgfplotsset{compat=newest}

\usepackage{mathrsfs}
\usepackage{algorithm}
\usepackage{algpseudocode}
\usepackage{url}
\usepackage{hyperref}
\usepackage{breakurl}
\hypersetup{breaklinks=true}
\usepackage[]{amsmath}
\usepackage{amsthm}
\usepackage{fix-cm}
\usepackage[]{amssymb}
\usepackage[]{latexsym}
\usepackage[right]{eurosym}
\usepackage[T1]{fontenc}
\usepackage[]{graphicx}
\usepackage[]{epsfig}
\usepackage{fancyhdr}
\usepackage{multirow}
\usepackage{subcaption}
\usepackage{wrapfig}
\usepackage[authoryear,round]{natbib}

 \usepackage{nopageno}

\allowdisplaybreaks

\makeatletter

\newcommand{\we}{\wedge}

\newcommand{\lm}{\lambda_{\max}}

\newcounter{modcount}
\newcommand{\modulo}[2]{%
\setcounter{modcount}{#1}\relax
\ifnum\value{modcount}<#2\relax
\else\relax
\addtocounter{modcount}{-#2}\relax
\modulo{\value{modcount}}{#2}\relax
\fi}
\newcommand{\tablepictures}[4][c]{\begin{tabular}[#1]{@{}c@{}}#2\vspace{0.5cm}\\(\alph{#4}) #3\end{tabular}}
\newcounter{gridsearch}
\newcommand{\tabpic}[2]{
    \stepcounter{gridsearch}
    \modulo{\thegridsearch}{2}
    \ifnum\value{modcount}=0
        \tablepictures[t]{#1}{#2}{gridsearch}\\[2.0cm]
    \else
        \tablepictures[t]{#1}{#2}{gridsearch}&~&
    \fi
}

\makeatother
\newtheorem{lemma}{Lemma}[section]
\newtheorem{proposition}[lemma]{Proposition}

\newtheorem{example1}[lemma]{Example}
\newtheorem{rem1}[lemma]{Remark}
\newtheorem{assumption}[lemma]{Assumption}
\newtheorem{me1}[lemma]{Mechanism}

\newenvironment{remark}{\begin{rem1}\rm}{\end{rem1}}
\newenvironment{example}{\begin{example1}\rm}{\end{example1}}

\usepackage{color}

\newcommand\ind[1]{\mathbb{I}_{\{#1\}}}

\definecolor{darkseagreen}{rgb}{0.56, 0.74, 0.56} 
\definecolor{carolinablue}{rgb}{0.6, 0.73, 0.89} 
\definecolor{columbiablue}{rgb}{0.61, 0.87, 1.0} 
\definecolor{coralpink}{rgb}{0.97, 0.51, 0.47} 
\definecolor{coral}{rgb}{1.0, 0.5, 0.31} 
\definecolor{darkchampagne}{rgb}{0.76, 0.7, 0.5}  
\definecolor{deepchampagne}{rgb}{0.98, 0.84, 0.65} 

\begin{document}

\title{
The not-so-hidden risks of `hidden-to-maturity' accounting: \\on depositor runs and bank resilience
}


\author{
Zachary Feinstein\thanks{Stevens Institute of Technology, School of Business, Hoboken, NJ 07030, USA. {\tt zfeinste@stevens.edu}.} 
\and
Grzegorz Ha{\l}aj\thanks{European Central Bank, Frankfurt, DE. {\tt grzegorz.halaj@ecb.europa.eu}.}
\and 
Andreas S{\o}jmark\thanks{London School of Economics, Department of Statistics, London, WC2A 2AE, UK. {\tt a.sojmark@lse.ac.uk}.}
}
\date{\today}
\maketitle

\begin{abstract}
We introduce a simple model of depositor runs to capture run risks at financial institutions based on their balance sheet composition. Specifically, we consider a reduced potential to raise capital from liquidity buffers under stress, during a stylized run driven by depositor scrutiny and further fueled by fire sales in response to withdrawals. The setup is inspired by the Silicon Valley Bank meltdown in March 2023 and concerns about the role of held-to-maturity portfolios. In particular, we apply our model to show the build-up of balance sheet vulnerabilities at Silicon Valley Bank before its default. More generally, the model may serve as a tool for analyzing which characteristics of the balance sheet are critical for banking system regulators to adequately assess run risk and resilience. Furthermore, we extend the model to include a tractable optimization problem which addresses the designation of held-to-maturity assets and provides a simple approach to quantifying banks' ability to hold those assets without resorting to remarking. To illustrate this framework, we calibrate the parameters to Silicon Valley Bank's balance sheet data and examine the bank's funding risk and implied risk tolerance in the years 2020--22 leading up to its collapse.


Keywords: bank runs, fire sales, accounting standards, funding risk
\end{abstract}


\section{Introduction}\label{sec:intro}

Despite a better understanding of bank runs and policies to mitigate them, run risks remain a significant concern. This was forcefully demonstrated by the March 2023 collapse of Silicon Valley Bank (SVB) and its serious repercussions across the broader financial markets. As discussed in the review of the Federal Reserve's supervision and regulation of SVB \cite{FRB2023}, two of the most significant internal factors were a high reliance on uninsured deposits and a stark increase in the proportion of assets designated as held-to-maturity (HtM). Based on this, the Vice Chair for Supervision at the Federal Reserve, Michael S.~Barr, concluded that \emph{`we should re-evaluate the stability of uninsured deposits and the treatment of held to maturity securities in our standardized liquidity rules and in a firm's internal liquidity stress tests'}.\footnote{See page 3 of https://www.federalreserve.gov/publications/files/svb-review-20230428.pdf} The goal of this paper is to take a step in that direction, by developing a simple theoretical framework that can help to shed light on and quantify how the level of uninsured deposits and the reliance on HtM securities portfolios impact run risks. Specifically, using the framework, we can assess the safe levels of the HtM portfolios that mitigate the run risk.

Portfolios held to maturity are long-term investments that are not meant to be a source of readily available liquidity. Since they are intended to be held to their maturity, accounting rules provide a cushion to banks' capital by not recognizing changes in the fair value of these portfolios whenever market conditions move. Instead, changes in risk factors create the so-called unrealized gains and losses in the HtM portfolios. When rates began to rise in Q1 of 2022, unrealized losses in SVB's HtM holdings grew rapidly, and the bank found itself limited in its ability to adjust its portfolio to the changing rate environment as selling part of the HtM securities would require its entire HtM holdings to be reclassified and marked-to-market \cite{FRB2023}. In other words, it appears SVB was relying too heavily on the HtM designation which is not meant for assets that may need to be sold early, even if only in unlikely events triggering sudden bank runs, i.e., material funding withdrawals. Indeed, for securities to be classified as HtM, the U.S.~GAAP rules require a bank to have the positive intent and ability to hold those securities to maturity.\footnote{See FASB ASC 320-10-25-1(c) which can be found at https://asc.fasb.org/1943274/2147481736} Both the very wording of these rules and the chain of events at SVB raises the question of how to sensibly assess such ability in relation to a bank's current funding structure and asset composition, when taking into account how these may interact with run risks. It should be noted that the IFRS9 accounting regime, adopted for instance by the European Union, has no explicit mention of a bank's ability to hold securities to maturity, but a specific recognition of assets must be commensurate with the bank's business model,\footnote{See https://www.ifrs.org/issued-standards/list-of-standards/ifrs-9-financial-instruments/. We also note that banks deciding to hold assets to collect payments, e.g., interest payments, would follow amortized cost accounting similar to the HtM option in GAAP. Moreover, EU banks' balance sheet sensitivity to interest rate shocks, including HtM portfolios, will be subject to additional supervisory scrutiny as part of the interest rate risk in banking book (IRRBB) limits, covered by the Supervisory Outlier Test (SOT) of the European Banking Authority (see a technical standard EBA/RTS/2022/10).} so similar considerations apply. Our model captures the accounting rules for HtM portfolios and thereby provides a link between these rules and bank run risks.

In a recent work, \cite{Granja2023} drew attention to the significant relative increase in banks' reliance on the HtM category during 2021--22, likely as a way of, at least partially, `hiding' potential future unrealized losses in interest rate sensitive assets. For exactly this type of reason, there has been a widespread public debate about the HtM framework and it has become known as a form of `hidden-to-maturity' accounting. Worryingly, the empirical findings of \cite{Granja2023} also suggest that the HtM accounting rules have been more frequently applied by less capitalized banks with significant uninsured deposits, thus prone to runs, presumably to immunize their capital from revaluation of securities held on-balance. Furthermore, \citet{Jiang2023_hedge} argue that reclassification of assets to HtM hides actual interest rate risk that is left unhedged and not adequately recognized by capital figures reported in financial statements. In times of distress, or simply when market expectations change, these issues can be detrimental to those banks' health and, as the March 2023 banking turmoil showed, to financial stability as a whole.

\subsection{Overview and main contributions}

In Section \ref{sec:bs}, we develop a stylized model to help explain bank runs by banks' financial conditions and how they are perceived by depositors. As part of this, we specify some basic assumptions that a bank's balance sheet should satisfy to allow us to derive bank run equilibrium conditions. These conditions are related to a critical level of leverage ratio that depositors consider as financially sound. Moreover, they describe some impact of asset liquidation on their value. 

Subsequently, in Section \ref{sec:clearing}, we define the clearing equilibrium for a run jointly in terms of deposit withdrawals and necessary asset liquidations. Proposition \ref{prop:exist} shows that there are minimal and maximal clearing solutions. Moreover, and important for computational tractability of the models, in Proposition \ref{prop:clearing_algo} we derive an algorithmic way to calculate the equilibrium values for a given bank and confirm that the algorithm converges. In this way, we provide a parsimonious theoretical framework that can serve to pinpoint the key drivers behind run-related instabilities, in particular identifying threshold ratios of banks’ balance sheets that delineate stable financial conditions from those conducive to bank runs. Importantly, these thresholds show whether a bank holds enough liquid assets or whether it needs to tap liquidity from HtM portfolios, whose primary purpose is to keep the securities until maturity or at least for the long term.

In Section \ref{sec:caseofsvb}, we illustrate the mechanics of the model based on two case studies. One is stylized, allowing us to demonstrate, in a controlled environment, how the equilibrium changes with the key parameters of the model: the leverage threshold and the price-impact sensitivity. The second example demonstrates that the model can be calibrated to the balance sheet data of a real bank, namely SVB, whose failure we discussed above.

In Section \ref{sect:AfS-HtM}, we finally formulate a tractable optimization problem as a way to explain how a prudent bank may choose its HtM designation, for a given level of marketable securities, when confronted with funding risk from a market price shock. Proposition \ref{prop:optim} provides a closed-form solution to the optimization problem and, consequently, with a procedure of how to compute the optimal HtM allocation. Thereby, we can address a bank's ability to hold different amounts of HtM securities without having to resort to remarking for reasonable price shocks. 

This methodology could be of relevance both for risk management practices, in particular internal stress tests, and for supervisory analysis through the monitoring of banks' commensurate use of the HtM classification. Although the model is highly stylized, it addresses a pertinent aspect of the HtM accounting rules that even forms part of the regulations as discussed above: namely, that banks should have not only the positive intent but also the ability to hold HtM securities until maturity. In Section \ref{app:beyondSVB}, we illustrate how the optimization can be used to detect financial vulnerabilities of banks related to their balance sheet structures based on the two case studies, i.e., with a stylized balance sheet structure and using calibrations to a time series of SVB's balance sheet data. The former shows how to check whether the level of HtM securities ensures sufficient liquidity to cover assumed shocks to depositors' preferences about the sound leverage ratio and shocks to the value of liquidity portfolios. The latter example enables us to analyze the build-up of vulnerabilities at the bank in the years ahead of its failure.

\subsection{Related literature}

Bank runs associated with risky projects and short-term, flighty funding have been studied extensively. Nevertheless, even with more stringent liquidity regulations and policy intervention frameworks in place, liquidity risk forcefully materializes time and again. Notably, in March 2023, SVB became a textbook example of a bank run, revitalising the discussion on bank-run risk \citep{Vo2023}.

In the classical literature, \cite{Diamond1983} first demonstrated that banks with illiquid assets can be susceptible to self-fulfilling runs, due to coordination failures, even when fundamentals are sound. A significant part of the later literature has been concerned with whether runs, and bank failures more generally, emerge from such sunspot phenomena or if they are driven by poor fundamentals \cite{Goldstein2013}. A key prediction of the global games approach \citep{Morris2003} is that runs occur only at low fundamentals, but they can be self-fulfilling nonetheless (when the fundamentals are not too low). If depositors have noisy observations of some true fundamental state, then there is a threshold value so that the (unique) equilibrium of the coordination game is a run precisely when the fundamentals are below this value \cite{GoldsteinPauzner2005}.

In this work, we are not looking to model why a run occurs. Instead, one of our model parameters is a maximum acceptable threshold for the bank's leverage ratio, and a run is then initiated if this is violated. Our interest lies in devising a tractable framework for computing and analyzing how the bank's balance sheet composition and exposure to fire sales affect the equilibrium outcome of a run.

To resolve a run, we formulate a clearing problem that matches withdrawal requests with cash raised from asset liquidations until the acceptable threshold is restored, if possible. This clearing point of view, based on the bank's balance sheet, is related to works on clearing and contagion in the systemic risk literature \cite{amini2016, veraart2020}. However, we stress that we do not consider financial networks. In a systemic context, \cite{DiamondRajan2005} have established the social inefficiency of runs and the consequent potential for contagion, as runs reduce the overall liquidity. To explore this further, one could formulate a network version of our model, but we do not pursue this herein.

We also stress that our framework incorporates the possibility of fire-sale externalities when a bank is forced to liquidate marketable but illiquid assets. Thus, our work is connected to the growing literature on fire sales as one of the key shock amplification channels when financial institutions need to raise cash by asset sales  \citep{Bichuch2019,Amini2025,Detering2022}.

A significant part of the paper is concerned with the role of the HtM accounting framework. The empirical findings of \cite{Granja2023} and \cite{Kim2023} indicate that neither the intent nor ability to hold HtM portfolios to maturity appear to be primary concerns of banks in general. This is based on datasets for when banks classify or re-classify their HtM holdings. As far as we are aware, there has been no attempt to analyze the reasonableness of HtM holdings within an optimization framework  that incorporates run risk. While intentions are not verifiable, our framework presents a simple tool to analyze and assess what constitutes sound levels of HtM holdings, thereby addressing a given bank's ability---and in turn the credibility of its intent---to hold assets to maturity. This furthermore links to work on the interface between accounting (or reporting) rules and financial stability. For example, \citep{Bischof2021} examine how prudential filters and accounting frameworks (like the amortized cost versus fair value accounting) feed into banks' incentives and thus financial stability. In this regard, we also mention the discussion of deposit freezes and redemption penalties in relation to bank runs in \citet{Altermatt2022}.

Finally, we highlight the recent work of \cite{Granja2024} which appeared while we were completing this paper. Arguably, their work is the closest to our analysis and general approach. The authors build on the bank run model of \citet{Jiang2023} and propose a maximization problem to study banks' designation of HtM assets. We discuss these two works in detail in Section \ref{subsect:recent_work_runs}. Here, we only note that they consider all-or-nothing runs related to insolvency, while we focus on characterizing various run outcomes, with funding and liquidity risk at the center of our model. In particular, our setup allows us to study the impact of fire sales and price-impact sensitivity of liquidated securities. To this end, we assume banks would sell outright securities from their portfolios instead of pledging those as collateral at the central bank or covering deposit outflows with emergency borrowing, likely at much higher cost, as documented by \citet{Cipriani2024} in their empirical work using payments data around SVB collapse. In this way, we can measure how banks can withstand funding shocks on a standalone basis.

\section{Balance sheet construction and model setup}\label{sec:bs}

We begin by specifying the stylized balance sheet of a given bank. The aim is to keep this as simple as possible, while allowing for enough granularity to capture key features of the distinctive roles played by the AfS and HtM designations as well as insured versus uninsured deposits in relation to bank stability. First of all, this will entail two classes of liabilities, i.e., deposits: they can be either insured $L_I \geq 0$ or uninsured $L_U \geq 0$ with the total liabilities given by $L=L_I+L_U$. Uninsured deposits are assumed to be flighty and subject to run risk, while insured depositors have no reason to run. However, the model could also be parameterized to reflect a different classification of liabilities, for instance, stable and unstable funding considered in FINREP reporting, since some of the insured depositors may \emph{`not want to continue to bank with a failing institution'}, see \citet{Cipriani2024}. In this broader perspective, only the split between stable funding ($L_I$) and runnable funding ($L_U$) matters for the model dynamics; for simplicity of calibration herein, we take the insured and uninsured deposits to represent these terms.

Next, we shall assume that the assets of the bank can be one of three types: liquid, illiquid but marketable, or illiquid and nonmarketable. In the case of illiquid but marketable assets, these may be classified as either available-for-sale or held-to-maturity. Beyond this classification, these illiquid but marketable securities are subject to the same market price. In summary, the asset side of the balance sheet will consist of the following four distinct classes:
\begin{enumerate}
\item liquid (cash) assets at mark-to-market value $x \geq 0$;
\item available-for-sale (AfS) illiquid assets $s \geq 0$ with an initial mark-to-market value of $sp$ for some unit price $p > 0$;
\item held-to-maturity (HtM) illiquid assets $h \geq 0$ valued in full (despite being subject to the same market price as the AfS assets); and
\item nonmarketable illiquid assets at book value $\ell \geq 0$.\footnote{Nonmarketable illiquid assets include also the franchise value of the bank. See also Section~\ref{sec:franchise}.}
\end{enumerate}
With the above notation, the total assets are given by $A=x + sp + h + \ell$. The bank's equity $E$ is then the difference between this value and the total liabilities, i.e., $E=A-L$. These quantities determine the a priori composition of the balance sheet before any considerations of a run by depositors. 

\begin{remark}
For now, we take the classification of AfS versus HtM as fixed and given. In Section \ref{sect:AfS-HtM}, we shall address this allocation through an optimization problem.
\end{remark}

Following \cite{banerjee2019b}, we assume that the illiquid, but marketable, holdings are subject to price impacts if they are sold. The mark-to-market value of these assets is given by the inverse demand function $f: [0,s+h] \to [0,p]$ for the initial price $p > 0$.
As these liquidations are performed, the bank faces the volume-weighted average price $ \bar f(\gamma) := \frac{1}{\gamma} \int_0^\gamma f(t)dt$, for $\gamma \in (0,s+h]$, with $\bar{f}(0):=p$. Note that $\bar{f}$ is continuous at $\gamma=0$ if $f$ is continuous there.
In this way, any unsold AfS assets are valued at the price determined by $f$, any sold AfS assets are valued at the price determined by $\bar f$, and any HtM assets are (initially) valued at a fixed price of $1$.
On the other hand, we will assume that the market value of the liquid assets ($x$) and the book value of the nonmarketable assets ($\ell$) remain fixed throughout this study.

\begin{figure}[H]
\centering
\begin{tikzpicture}
\draw[draw=none] (0,6.5) rectangle (6.5,7) node[pos=.5]{\bf A Priori Balance Sheet};
\draw[draw=none] (0,6) rectangle (3.25,6.5) node[pos=.5]{\bf Assets};
\draw[draw=none] (3.25,6) rectangle (6.5,6.5) node[pos=.5]{\bf Liabilities};

\filldraw[fill=darkseagreen,draw=black] (0,5.3) rectangle (3.25,6) node[pos=.5,style={align=center}]{Liquid $x$}; 
\filldraw[fill=carolinablue,draw=black] (0,2.5) rectangle (3.25,5.3) node[pos=.5,style={align=center}]{Available for Sale \\ $sp$};
\filldraw[fill=columbiablue,draw=black] (0,1.25) rectangle (3.25,2.5) node[pos=.5,style={align=center}]{Held to Maturity \\ $h$};
\filldraw[fill=coralpink,draw=black] (0,0) rectangle (3.25,1.25) node[pos=.5,style={align=center}]{Nonmarketable \\ $\ell$};

\filldraw[fill=coral,draw=black] (3.25,3) rectangle (6.5,6) node[pos=.5,style={align=center}]{Insured Deposits \\ $L_I$};
\filldraw[fill=darkchampagne,draw=black] (3.25,1.5) rectangle (6.5,3) node[pos=.5,style={align=center}]{Uninsured Deposits \\ $L_U$};
\filldraw[fill=deepchampagne,draw=black] (3.25,0) rectangle (6.5,1.5) node[pos=.5,style={align=center}]{Equity};

\draw[->,line width=1mm] (7,3) -- (8,3);

\draw[draw=none] (8.5,6.5) rectangle (15,7) node[pos=.5]{\bf Realized Balance Sheet};
\draw[draw=none] (8.5,6) rectangle (11.75,6.5) node[pos=.5]{\bf Assets};
\draw[draw=none] (11.75,6) rectangle (15,6.5) node[pos=.5]{\bf Liabilities};

\filldraw[fill=darkseagreen,draw=black] (8.5,4.8889) rectangle (11.75,6) node[pos=.5,style={align=center}]{Liquid \\ $x+\gamma \bar f(\gamma)$};
\filldraw[fill=carolinablue,draw=black] (8.5,2.6444) rectangle (11.75,4.8833) node[pos=.5,style={align=center}]{Available for Sale \\ $(s-\gamma) f(\gamma)$};
\filldraw[fill=columbiablue,draw=black] (8.5,1.3944) rectangle (11.75,2.6444) node[pos=.5,style={align=center}]{Held to Maturity \\ $h$};
\filldraw[fill=coralpink,draw=black] (8.5,0.1444) rectangle (11.75,1.3944) node[pos=.5,style={align=center}]{Nonmarketable \\ $\ell$};
\filldraw[fill=carolinablue,draw=black] (8.5,0) rectangle (11.75,0.1444);

\filldraw[fill=coral,draw=black] (11.75,3) rectangle (15,6) node[pos=.5,style={align=center}]{Insured Deposits \\ $L_I$};
\filldraw[fill=darkchampagne,draw=black] (11.75,2.6111) rectangle (15,3) node[pos=.5,style={align=center}]{Uninsured Deposits};
\filldraw[fill=darkchampagne,draw=black] (11.75,1.5) rectangle (15,2.6111) node[pos=.5,style={align=center}]{Withdrawals \\ $w$};
\filldraw[fill=deepchampagne,draw=black] (11.75,0.1444) rectangle (15,1.5) node[pos=.5,style={align=center}]{Equity};
\filldraw[fill=deepchampagne,draw=black] (11.75,0) rectangle (15,0.1444);
\draw (8.5,0) rectangle (15,6);
\draw (11.75,0) -- (11.75,6);

\begin{scope}
    \clip (8.5,0) rectangle (15,0.1389);
    \foreach \x in {-8.5,-8,...,15}
    {
        \draw[line width=.5mm] (8.5+\x,0) -- (15+\x,6);
    }
\end{scope}
\begin{scope}
    \clip (8.5,4.8889) rectangle (11.75,6);
    \foreach \x in {-8.5,-8,...,15}
    {
        \draw[line width=.5mm,color=gray] (8.5+\x,0) -- (15+\x,6);
    }
\end{scope}
\begin{scope}
    \clip (11.75,1.5) rectangle (15,2.6111);
    \foreach \x in {-8.5,-8,...,15}
    {
        \draw[line width=.5mm,color=gray] (8.5+\x,0) -- (15+\x,6);
    }
\end{scope}
\end{tikzpicture}
\caption{Stylized balance sheet before and after a run given by an amount $w$ of deposit withdrawals. In this case, a quantity $\gamma$ of AfS assets had to be sold, while the HtM assets were left untouched.}
\label{fig:balance_sheet}
\end{figure}

The initial composition of the balance sheet and an example of the outcome after a run are illustrated in Figure~\ref{fig:balance_sheet}. A depositor run is resolved as follows. Based on the observed balance sheet, the \emph{uninsured} depositors will withdraw funds depending on whether or not the leverage ratio, given by assets over equity, is in line with some maximum acceptable threshold. If violated, withdrawals take place until the realized leverage ratio is back in line with the threshold. We denote by $w \in [0,L_U]$ the total withdrawals that result from this and we let $\gamma$ denote the corresponding quantity sold of the marketable assets.

To be precise, the uninsured investors have a maximum leverage ratio $\lm>1$ that they are willing to accept. If the actual leverage ratio is above this value, a run is initiated. The resulting equilibrium outcome $(w^*,\gamma^*)$ is one such that the withdrawal requests $w^*$ bring the actual (or realized) leverage ratio $\lambda=\lambda(w^*,\gamma^*)$ back in line with $\lm$ given the new mark-to-market values implied by the quantity $\gamma^*$ sold to cover withdrawals. Here the actual leverage ratio $\lambda=\lambda(w,\gamma)$ is defined as the ratio of assets over equity when accounting for withdrawals and selling as well as the corresponding losses that must be recognized on the balance sheet. This means that, for their decision to withdraw or not, depositors look at the realized mark-to-market values of all marketable assets (in determining the leverage ratio), while factoring in that HtM assets need only be counted at market values if withdrawals are large enough that part or all of them will need to be sold. In the case that HtM assets are remarked to be sold, we impose a dead-weight loss, proportional to the size of the assets, on the bank. Naturally, our model is highly stylized in that the resulting equilibrium is solved for in a single step, but one should have in mind continuously occurring withdrawals up until the equilibrium is reached. Finally, we recall that \emph{insured} depositors are assumed to leave their funds at the bank even in a stress scenario---in particular, a bank run---due to the guarantee of recovery in case of a failure. 

Even though we focus on how a bank run is resolved and we do not endogenize the signals for a bank run to start, we provide a motivation for how it can get in motion. Once the bank's fundamentals, simplistically measured by leverage ratio, are perceived to indicate material deterioration of this bank's financial soundness, impatient uninsured depositors would start withdrawing. They would observe the bank's ability to raise cash and how the financial indicators would evolve, esp., to assess the bank's loss absorption capacity on a going concern. They could stop the run if the depositors do not overreact and bank's solvency conditions meet their risk tolerance, in which case it could be a partial run.

We use the leverage ratio as the decision variable for depositors, since it allows for a tractable analysis and since it is one of the key financial indicators to determine bank stability, allowing investors and depositors to track changes in solvency risk, esp.\ useful in case of uncertainty of asset valuation \citep{Dermine2015}.\footnote{It is one of several key financial stress indicators \citep{DUCA2013}.} In particular, it has been used as an important variable in seminal bank run models \citep{Gertler2015}. Furthermore, the leverage ratio is one of two key solvency indicators that is regulated by capital standards, most importantly by Basel III regulation introduced after the Global Financial Crisis. Banks should keep it above the regulatory minimum and typically retain a voluntary buffer, so as to minimize the risk that the leverage ratio falls below requirements.\footnote{See \url{https://www.bis.org/basel\_framework/standard/LEV.htm}} Naturally, investors could also look at other indicators, related to the funding and liquidity position of the banks, however, those are more difficult to track, as they are reported with a considerable lag, and the investors may want to react to more representative signals coming from the solvency angle. In this regard, see also Remark \ref{rem:barr} below. Moreover, we want to keep the withdrawal decisions anchored in our stylized balance sheet description of the bank. For these reasons, we focus on the leverage ratio as the sole signal tracked by depositors.

\begin{remark}\label{rem:barr}
While bank runs are typically associated with liquidity issues, we take the view of Michael S.\ Barr, Federal Reserve Board Vice Chair for Supervision, that `\emph{while the proximate cause of SVB’s failure was a liquidity run, the underlying issue was concern about its solvency}.'\footnote{See page 2 of \url{https://www.federalreserve.gov/publications/files/svb-review-20230428.pdf}} 
In this way, solvency concerns of uninsured depositors can manifest as liquidity problems for the bank. In our model, the maximum acceptable leverage ratio may be viewed as the threshold for when worries about solvency kick in. As documented by \citet{Fascione2024}, though there are many drivers of depositors' decisions to move their deposits, depositor outlook on solvency is one of the key concerns.
\end{remark}

Whilst our model is static, it can eventually be applied for bank balance sheets at multiple points in time. Thus, it is worth mentioning that the acceptable maximum leverage ratio $\lm$ could in principle be changing in time:~it may be a function of changing macro-financial environments or changing risk tolerances of the depositors. For example, risk aversion might have changed in 2020 due to wealth effects (or expected wealth effects) implied by COVID-19 crisis (see drivers of risk aversion studied by \citet{Guiso2018}), while in 2021-2022 it might have been more stable given no new shocks with magnitudes comparable to the pandemic. One could then think of sudden drops in $\lm$ as the reason for a run, spurred by changing depositor sentiments, but we shall not pursue such an angle here. Rather, we will view $\lm$ as a given, albeit unobservable, characteristic of depositors that remains fixed as the balance sheet evolves over some period of time of interest. Thereby, changes in run risk are explained by observable changes to the bank's balance sheet composition for given characteristics of the depositors. In particular, when we perform quarterly simulations based on a time series of SVB's balance sheet data in Section \ref{sec:caseofsvb}, we assume that $\lm$ stays constant throughout. That way, we can study vulnerabilities of SVB exclusively stemming from its evolving balance sheet structure. Naturally, this will involve inferring plausible levels for $\lm$, and we will consider a range of such values for some of the simulations. In Section \ref{app:beyondSVB}, we also illustrate how to apply the model outcomes in relative terms comparing vulnerabilities across banks which require considering different potential levels of acceptable leverage ratios.

\begin{assumption}\label{ass:idf}
The inverse demand function $f: [0,s+h] \to (0,p]$ is non-increasing with initial price $f(0) = p$, where $p\in(0,1]$.
\end{assumption}

\begin{assumption}\label{ass:LU}
We assume $L_U > 0$ as no withdrawals would occur otherwise.
\end{assumption}

For modeling purposes, we stress that the quantities $L_I$, $L_U$, and $L$ remain fixed, as they capture the given liabilities of the initial balance sheet, before a run. The uninsured liabilities after withdrawals are then given by $L_U-w$. Since nothing is withdrawn from the insured liabilities $L_I$, the total liabilities thus become $L-w$. Writing $A=A(w,\gamma)$ for the total assets (with recognized losses) as a function of the withdrawals $w$ and the quantity of marketable securities sold $\gamma$, we can express the actual (realized) leverage ratio $\lambda=\lambda(w,\gamma)$ as
\begin{equation}\label{eq:lev_ratio}
\lambda = \frac{A}{E}= \frac{A(w,\gamma)}{A(w,\gamma)-(L-w)},
\end{equation}
where
\[
A(w,\gamma)= x + \gamma \bar f(\gamma) + ([s-\gamma]f(\gamma) + h)\ind{\gamma \leq s} + [s+\alpha h-\gamma]f(\gamma)\ind{\gamma > s} + \ell - w,
\]
for the given values of $x$, $s$, $h$, and $\ell$ and where $\alpha \in (0,1]$ provides a cost to remarking, i.e., $(1-\alpha)h$ are dead-weight losses from the act of converting from HtM to AfS.

If the withdrawal requests are larger than the bank's liquid holdings, then it will need to raise cash by selling its illiquid, but marketable, asset holdings. Specifically, the bank must sell $\gamma \in [0,s+\alpha h]$ so that $x + \gamma \bar f(\gamma) \geq w$, if possible.\footnote{We impose a no short selling constraint so that $\gamma \leq s+\alpha h$ throughout.} Based on the given balance sheet, if $\gamma \leq s$ then we assume that all liquidated assets were AfS assets. On the other hand, if $\gamma > s$ then the bank liquidates all AfS assets and must be liquidating HtM assets as well. Formally, if any HtM assets are to be liquidated, then that entire block of assets is immediately recognized as AfS and marked accordingly with the associated dead-weight losses applied. Summarizing these two cases:
\begin{enumerate} 
\item if $\gamma \leq s$ then the bank holds $x + \gamma \bar f(\gamma)$ in liquid assets, $(s-\gamma)f(\gamma)$ in AfS assets, and $h$ in HtM assets;
\item if $\gamma > s$ then the bank holds $x + \gamma \bar f(\gamma)$ in liquid assets, $(s+\alpha h-\gamma)f(\gamma)$ in AfS assets, and $0$ HtM assets.
\end{enumerate}

However, in satisfying the withdrawals, the bank may fail due to having insufficient liquidity or insufficient equity.  We call these cases illiquidity and insolvency respectively. 
\begin{enumerate}
\item \textbf{Illiquidity}: The bank cannot meet withdrawals: $w \geq x + [s+\alpha h] \bar f(s+\alpha h)$ or, equivalently, $\gamma = s+\alpha h$.
\item \textbf{Insolvency}: The bank has negative equity: \[x + \gamma \bar f(\gamma) + ([s-\gamma]f(\gamma) + h)\ind{\gamma \leq s} + [s+\alpha h-\gamma]f(\gamma)\ind{\gamma > s} + \ell \leq L.\] 
\end{enumerate}
\begin{remark}
If there are no price impacts on the illiquid asset, i.e., $f \equiv p$, then insolvency can only occur at the moment that the HtM assets are re-marked as AfS assets.
\end{remark}

For our analysis, we shall need a final assumption on the behavior of the realized balance sheet.  Recall that $\lm>1$. As a function of the quantity sold, we require that the rate of increase in the realized value of the (total) assets sold is always larger, by a factor of $\frac{\lm -1}{\lm}=1-\frac{1}{\lm}>0$, than the corresponding rate of decrease in the market value of the remaining unsold assets. More precisely, we impose the following technical assumption on the inverse demand function $f$.

\begin{assumption}\label{ass:idf-2}
For the remainder of this work, we will assume that the mapping $\gamma \in [0,s+\alpha h] \mapsto \gamma \bar f(\gamma) + (1 - \frac{1}{\lm})(s+\alpha h-\gamma)f(\gamma)$ is strictly increasing.
\end{assumption}

\begin{lemma}\label{lem:inverse_deman_assump} Suppose the inverse demand function $f$ is differentiable on $(0,s+\alpha h)$. Then Assumption~\ref{ass:idf-2} holds provided the differential inequality
\[ f(\gamma) > (1-\lm)(s+\alpha h-\gamma) f^\prime(\gamma)
\]
is satisfied, for all $\gamma \in [0,s+\alpha h]$.
\end{lemma}
\begin{proof} 
It suffices to check that $\gamma \mapsto \gamma \bar f(\gamma) + (1 - \frac{1}{\lm})(s+\alpha h-\gamma)f(\gamma)$ has a strictly positive derivative on $(0,s+\alpha h)$. Using the definition of $\bar{f}$, differentiating, and reorganizing the terms, we see that this is equivalent to the stated differential equality.
\end{proof}

\begin{remark}\label{rem:increas_non-increase}
Under Assumption \ref{ass:idf-2}, we get that  $\gamma \mapsto \gamma \bar f(\gamma) + (1-\frac{1}{\lm})(\bar s - \gamma)f(\gamma)$ is strictly increasing on $ [0,\bar s]$, for any $\bar s \in (0,s+\alpha h]$. In particular, this holds for $\bar s \in \{s, s+\alpha h \}$ which we shall make use of in Section \ref{sec:clearing}. At the same time, we stress that the map $\gamma \mapsto \gamma \bar f(\gamma) + (\bar s - \gamma)f(\gamma)$ is instead non-increasing on $ [0,\bar s]$, for any $\bar s \in (0,s+\alpha h]$, as one can readily verify by, e.g.,~arguing as in Lemma \ref{lem:inverse_deman_assump}.
\end{remark}

We conclude this section by highlighting two common examples of inverse demand functions and outlining the parameter choices for which our assumptions are satisfied.

\begin{example}
Take $f(\gamma) := p(1 - b\gamma)$ as in, e.g.,~\cite{Greenwood2015}. Then $\bar f(\gamma) = p(1 - \frac{b}{2}\gamma)$. Naturally, $b \leq 1/(s+\alpha h)$ and hence
Assumptions~\ref{ass:idf} is satisfied. By Lemma \ref{lem:inverse_deman_assump}, Assumption \ref{ass:idf-2} holds if and only if either $b < 1/(s+\alpha h)$ for $\lm\in (1,2)$ or $b<1/[(\lm-1)(s+\alpha h)]$ for $\lm\geq 2$.
\end{example}
\begin{example}
Take $f(\gamma) = p\exp(-b\gamma)$ as in, e.g.,~\cite{Cifuentes2005}. Then $\bar f(\gamma) = \frac{p(1 - \exp(-b\gamma))}{b\gamma}$ for $\gamma > 0$ and $\bar f(0) = p$. Naturally, $b\geq 0$, so Assumption~\ref{ass:idf} holds. By Lemma \ref{lem:inverse_deman_assump},
Assumption~\ref{ass:idf-2} holds if and only if $b < 1/ [(\lm-1)(s+\alpha h)]$.
\end{example}

\section{The clearing problem for a depositor run}\label{sec:clearing}

In the previous section, we laid out the granularity of the balance sheet and introduced the mechanisms behind our bank run model. The goal of the present section is to first (i) formalize this as a precise clearing problem, then (ii) establish the existence of equilibrium solutions to this clearing problem, and finally (iii) provide a tractable algorithm for computing the equilibrium values. Throughout, we are working with the notation and under assumptions presented in Section~\ref{sec:bs}. 

Our model for a (potential) bank run may be expressed as the solution to a clearing problem that is jointly in the equilibrium amount of withdrawals $w^*$ and the equilibrium quantity sold $\gamma^*$ out of the marketable securities. This may be formalized as the search for fixed points of the mapping $\Phi: [0,L_U] \times [0,s+\alpha h] \to [0,L_U] \times [0,s+\alpha h]$ defined by $\Phi=(\Phi_w,\Phi_\gamma)$, where
\begin{align}
\Phi_w(\gamma^*) &= L_U \wedge \Bigl[\lm L - (\lm-1)(x + \gamma^* \bar f(\gamma^*) + ([s-\gamma]f(\gamma) + h)\ind{\gamma \leq s} \label{eq:map_w}\\
 &\qquad + [s+\alpha h-\gamma]f(\gamma)\ind{\gamma > s} + \ell)\Bigr]^+ \nonumber\\
\Phi_\gamma(w^*,\gamma^*) &= [s+\alpha h] \wedge \frac{(w^*-x)^+}{\bar f(\gamma^*)} \label{eq:map_gamma}.
\end{align}
For a given quantity sold, \eqref{eq:map_w} returns the withdrawals required for the depositors to enforce their maximum acceptable leverage ratio. Given also the withdrawals, \eqref{eq:map_gamma} then ensures that the proceeds from the quantity sold match the withdrawal requests. A pair $(w^*,\gamma^*)\in [0,L_U]\times [0,s+\alpha  h]$ is therefore a clearing solution if and only if it is a fixed point of $\Phi$, meaning that we have
\begin{equation}\label{eq:fixed_point}
(w^*,\gamma^*) = \Phi(w^*,\gamma^*) = (\Phi_w(\gamma^*) \; , \; \Phi_\gamma(w^*,\gamma^*)),
\end{equation}
provided also that the bank is solvent in this case, i.e., provided
\begin{equation}\label{eq:solvent}
x + \gamma^* \bar f(\gamma^*) + ([s-\gamma]f(\gamma) + h)\ind{\gamma \leq s} + [s+\alpha h-\gamma]f(\gamma)\ind{\gamma > s} + \ell > L. 
\end{equation}

If $(w^*,\gamma^*)$ satisfies \eqref{eq:fixed_point}, but violates \eqref{eq:solvent}, then the bank is insolvent. In that case, the values $(w^*,\gamma^*)$ correspond to the run having occurred and the bank only subsequently being declared insolvent. This is arguably more in line with the timeline of events in an actual run, but one can of course also look for the amount of liquidations $\gamma$ that first induces technical insolvency by violating \eqref{eq:solvent} during the run.

 For clearing solutions corresponding to a run (i.e., $w^*> x$) without causing illiquidity (i.e., $\gamma^* < s+\alpha  h$), we have $w^* = x + \gamma^* \bar{f}(\gamma^*)$ with all withdrawals being met, solvency issues aside. On the other hand, illiquidity corresponds to clearing solutions with a quantity sold $\gamma^*=s+\alpha h$ and withdrawals $w^*\geq x + (s+\alpha h) \bar{f}(s+\alpha h)$. When a bank is left illiquid, the value of $w^*$ reflects the withdrawal requests and not the actualized withdrawals, as the bank would generally not be able to cover all requests.

\begin{proposition}[Minimal and maximal clearing solutions]\label{prop:exist}
Consider the partial order of component-wise inequality. Then, there exist minimal and maximal clearing solutions $ (w^\downarrow,\gamma^\downarrow) \leq (w^\uparrow,\gamma^\uparrow)$.
\end{proposition}

Throughout, we work with the minimal solution, as this is the best case for the bank and represents the solution that a run would most naturally arrive at, without unnecessarily pessimistic views on what an outcome might look like. The next result provides an easy-to-use algorithm (Algorithm \ref{alg:clearing}) for finding the minimal clearing solution. When the algorithm terminates, one must additionally confirm whether or nor the candidate clearing solution leaves the bank solvent.

\begin{proposition}[Step-by-step algorithm]\label{prop:clearing_algo}
The minimal clearing solution $(w^\downarrow,\gamma^\downarrow)$ may be obtained by proceeding stepwise through Algorithm \ref{alg:clearing} until a conditional statement is true. If the resulting pair $(w^\downarrow,\gamma^\downarrow)$ satisfies \eqref{eq:solvent}, then the bank is solvent. If instead \eqref{eq:solvent} is violated, and the algorithm terminated before Step 4, then the bank is liquid but insolvent. If \eqref{eq:solvent} is violated and the algorithm terminated in Step 4, then the bank is both illiquid and insolvent.
\end{proposition}

The reasoning behind this result is as follows. We can consider fictitious runs in which we assume there is either (1) no selling of AfS assets, (2) some selling of AfS assets, but no HtM assets are sold, (3) HtM assets are sold, but enough to satisfy withdrawals, or (4) all HtM assets are sold and the bank is left illiquid. These are precisely the four steps of Algorithm \ref{alg:clearing} with the proof of Proposition \ref{prop:clearing_algo} providing the given characterizations. Now, we know from Proposition \ref{prop:exist} that there is a minimal clearing solution. Therefore, as the steps we have outlined are in increasing severity, the first one to yield a clearing solution must be the minimal such solution.

\definecolor{state1color}{RGB}{0, 150, 0}    
\definecolor{state2color}{RGB}{255, 204, 0}  
\definecolor{state3color}{RGB}{255, 0, 0}    
\definecolor{state4color}{RGB}{0, 0, 0}      

\begin{algorithm}
	\caption{Clearing Algorithm}\label{alg:clearing}
	\begin{algorithmic}[1]
		\State \textbf{No sales} \hfill \textcolor{state1color}{\rule[0.5ex]{2em}{3pt}}
		\vspace{0.2em}
		\Statex \hspace{1em} \textbf{If} $L_U \leq x$ or $\lm L - (\lm-1)(x+sp+h+\ell) \leq x$, \textbf{then}
		$\gamma=0$ and  $w=L_U \wedge [\lm L - (\lm-1)(x+sp+h+\ell)]^+$.
		
		\vspace{0.5em}
		\State \textbf{Run without HtM re-marking} \hfill \textcolor{state2color}{\rule[0.5ex]{2em}{3pt}}
		\vspace{0.2em}
		\Statex \hspace{1em} \text{(i)} \emph{Partial withdrawals}: \textbf{if}\vspace{-1pt}
		\begin{align*}
			&L-x-(1-\tfrac{1}{\lm})(h+\ell) \in [(1-\tfrac{1}{\lm})sp , s\bar f(s)], \quad \text{and} \\
			&L_U \geq \lm L - (\lm-1)(x+\gamma^*\bar f(\gamma^*) + (s-\gamma^*)f(\gamma^*)+h+\ell),\quad \text{for}\\
			&\gamma^* \bar f(\gamma^*) + (1-\tfrac{1}{\lm})(s-\gamma^*)f(\gamma^*) = L - x - (1-\tfrac{1}{\lm})(h+\ell), \quad \gamma^*\in[0,s],
		\end{align*}\vspace{-1pt}
		\hspace{1em} \textbf{then} $\gamma = \gamma^*$ and $w = x+ \gamma^*\bar{f}(\gamma^*) \in (x,L_U)$.
		
		\vspace{0.3em}
		\Statex \hspace{1em} \text{(ii)} \emph{Full withdrawals}: \textbf{if}
		$L_U \in (x , x+s\bar f(s)]$ and
		$L_I \geq (1 - \frac{1}{\lm})[(s-\gamma^*)f(\gamma^*) + h + \ell]$ for $\gamma^*\in[0,s]$ solving $\gamma^* \bar f(\gamma^*) = L_U - x$, \textbf{then}
		$\gamma = \gamma^*$ and $w = L_U$.
		
		\vspace{0.5em}
		\State \textbf{Run with HtM re-marking} \hfill \textcolor{state3color}{\rule[0.5ex]{2em}{3pt}}
		\vspace{0.2em}
		\Statex \hspace{1em} \text{(i)} \emph{Partial withdrawals}: \textbf{if}\vspace{-1pt}
		\begin{align*}
			&L-x-(1-\tfrac{1}{\lm})\ell \in [s\bar f(s) + (1-\tfrac{1}{\lm})\alpha hf(s) , (s+\alpha h)\bar f(s+\alpha h)],\quad \text{and} \\
			& L_U \geq \lm L - (\lm-1)(x+\gamma^*\bar f(\gamma^*) + (s+\alpha h-\gamma^*)f(\gamma^*) + \ell),\quad \text{for} \\
			&\gamma^* \bar f(\gamma^*) + (1-\tfrac{1}{\lm})(s+\alpha h-\gamma^*)f(\gamma^*) = L-x-(1-\tfrac{1}{\lm})\ell,\quad \gamma^*\in[s,s+\alpha h],
		\end{align*}\vspace{-1pt}
		\hspace{1em} \textbf{then} $\gamma = \gamma^*$ and $w = x + \gamma^*\bar f(\gamma^*)\in(x,L_U)$.
		
		\vspace{0.2em}
		\Statex \hspace{1em} \text{(ii)} \emph{Full withdrawals}: \textbf{if} $L_U \in (x , x+(s+\alpha h)\bar f(s+\alpha h)]$ and
		$L_I \geq (1 - \frac{1}{\lm})[(s+\alpha h-\gamma^*)f(\gamma^*)+\ell]$ for $\gamma^*\in[s,s+\alpha h]$ solving $\gamma^* \bar f(\gamma^*) = L_U - x$, \textbf{then} $\gamma = \gamma^*$ and $w = L_U$.
		
		\vspace{0.5em}
		\State \textbf{Illiquidity} \hfill \textcolor{state4color}{\rule[0.5ex]{2em}{3pt}}
		\vspace{0.2em}
		\Statex \hspace{1em} \text{(i)} \emph{Partial withdrawals}:  \textbf{if} \vspace{-4pt}
		$$\lm L - (\lm - 1)(x + (s+\alpha h)\bar f(s+\alpha h) + \ell) < L_U\;\; \text{and} \;\; L \geq x + (s+\alpha h)\bar f(s+\alpha h) + (1 - \tfrac{1}{\lm})\ell, \vspace{-4pt}$$
		\textbf{then} $\gamma = s+\alpha h$ and $w = \lm L - (\lm - 1)(x + (s+\alpha h)\bar f(s+\alpha h) + \ell)\in(x,L_U)$. 
		
		\vspace{0.3em}
		\Statex \hspace{1em} \text{(ii)} \emph{Full withdrawals}: \textbf{if} \vspace{-4pt}
		$$	\lm L - (\lm - 1)(x + (s+\alpha h)\bar f(s+\alpha h) + \ell) \geq L_U\;\; \text{and}\;\; L_U -x \geq (s+\alpha h)\bar f(s+\alpha h) \vspace{-4pt}$$
		\textbf{then} $\gamma = s+\alpha h$ and $w =L_U$.
	\end{algorithmic}
\end{algorithm}

Intuitively, Algorithm \ref{alg:clearing} describes the different `shades of liquidity' of a bank facing funding shocks while locked in with a given balance sheet of HtM securities, and it provides a simple step-by-step way to compute the equilibrium of deposit withdrawals and liquidation of securities. The shades refer to the four steps: from the most sound balance sheet to the weakest one that threatens with a default on illiquidity grounds. The first step is a `no sale' liquidity zone. It means that any deposit withdrawal shock can be covered with cash holdings. In the formulas of the proposition, we rewrite the constraints to have a leverage ratio more straightforwardly comparable with the maximum acceptable level. In this way, we can see that the more cash $x$ is held, the larger the distance to the maximum acceptable leverage. In the second step, withdrawals are at the maximum level of $L_U$, provided that $L_U$ exceeds cash $x$ and that the insured deposits surpass the AfS portfolio that can be cashed in together with the HtM portfolio and the other assets $\ell$. Subcases (i) and (ii) differ only in whether the withdrawals hit their maximum possible level $L_U$. In step 3, the bank must enter HtM, which triggers a remarking of the entire HtM portfolio. Finally, the last step 4 means that the bank has entered an illiquid situation, with insufficient resources to cover funding withdrawals. In all steps, the key driving forces are the maximum acceptable leverage, the hit from HtM remarking, and the price impact function that would determine how much cash would eventually be raised from the liquidated assets.

\begin{remark}\label{rem:all_solutions}
Though Algorithm~\ref{alg:clearing} presents a sequential procedure to compute the minimal clearing solution $(w^\downarrow,\gamma^\downarrow)$, the logic therein can be used to find all feasible clearing solutions $(w^*,\gamma^*)$. Specifically, rather than sequentially searching for a fixed point, every case within Algorithm~\ref{alg:clearing} is independently checked, i.e., the relevant condition is satisfied and that case admits the relevant clearing solution or the condition fails and no solution exists in that case. In particular, by running Algorithm~\ref{alg:clearing} backwards, the maximal clearing solution $(w^\uparrow,\gamma^\uparrow)$ can be found via an analogous proof to that of Proposition~\ref{prop:clearing_algo}. 
\end{remark}

The above results are concerned with positive initial equity $E(0,0)>0$. If, instead, a bank starts with non-positive equity $E(0,0)\leq 0$, then it is understood that a run is initiated (any positive maximum acceptable leverage ratio is violated), and one readily sees that the outcome is a full run, i.e., the withdrawal requests are $w^*=L_U$. Indeed, $E(0,0)\leq 0$ implies $\lm L-(\lm-1)A(0,\gamma^*)\geq L$, so we have withdrawal requests $\Phi_w(\gamma^*)=L_U$ for any quantity sold $\gamma^*\geq0$ in equation\eqref{eq:map_w}.

\subsection{Franchise value}\label{sec:franchise}

In addition to our model above, it can be relevant to consider a positive franchise value on the bank's balance sheet, due to market power in the deposit market \cite{Neumark1992,Drechsler2017}. Franchise value is generally understood as the expected present value of future profits from the ongoing deposit business, or some similar quantity. For our main analysis, we focus on the conservative case where this is zero, but here we show how a positive value can be incorporated.

Since franchise value cannot be liquidated (or used to raise cash) in a run, we can include it in the nonmarketable category of our balance sheet. Net of costs, we let it be given by a value $\delta>0$ per unit deposit. Consequently, the assets with franchise value become $A^\delta(w,\gamma):=A(w,\gamma)+\delta (L-w)$, where $A(w,\gamma)$ is the $\delta=0$ baseline considered above. The equity is $E^\delta(w,\gamma)=A^\delta(w,\gamma)-(L-w)$. Importantly, the total franchise value decreases with withdrawals and cannot be used as liquidity.

We assume $E^\delta(0,0)>0$ (otherwise, there is always a full run). As in \eqref{eq:lev_ratio}, the realized leverage ratio is then defined by  $\lambda^\delta(w,\gamma)= A^\delta(w,\gamma)/E^\delta(w,\gamma)$, and a run is initiated if $\lambda^\delta(0,0)>\lm$. We set $\kappa:=(\lm-1)\delta$, where we recall that $\lm > 1$.

When the franchise value is large relative to $\lm$, we observe the following. 

\begin{proposition}[Large franchise value]\label{prop:franchise_1}
Let $\kappa\geq1$. If the bank is solvent even in the absence of its franchise value (i.e., $E(0,0)\geq 0$), then we must have $\lambda^\delta(0,0) \leq \lm$, so there is never a run. If the bank is insolvent without its franchise value (i.e., $E(0,0) < 0$), then it is possible to have $ \lambda^\delta(0,0) > \lm$ and, in that case, the run is necessarily a full run (i.e., $w^*=L_U$).
\end{proposition}

If the franchise value is smaller, we can have partial runs, which are characterized as in Algorithm~\ref{alg:clearing}.

\begin{proposition}[Small franchise value]\label{prop:franchise_2}
Let $\kappa<1$. If $E(0,0)>0$, a run initiated by $\lambda^\delta(0,0)>\lm$ is resolved as follows. Set $\hat\lambda_{\max}:=\tfrac{\lm-\kappa}{1-\kappa}$ and suppose Assumption~\ref{ass:idf-2} holds with the more lenient threshold $\hat\lambda_{\max}\geq \lm$ in place of $\lm$. Then, the clearing solutions are given by Proposition~\ref{prop:exist} with $\hat\lambda_{\max}$ in place of $\lm$, and Algorithm~\ref{alg:clearing} applies with this substitution and the adjusted solvency condition $A(0,\gamma^*)+\delta(L-w^*)>L$. If $E(0,0)\leq 0$, there is a full run on the bank (i.e., $w^*=L_U$).
\end{proposition}

In the final Sections \ref{sec:caseofsvb} and \ref{sect:AfS-HtM}, we shall work with the conservative case of zero franchise value, relying on Algorithm \ref{alg:clearing}. A positive $\delta$ can be incorporated by using the two above propositions.

\subsection{Related recent work on bank runs}\label{subsect:recent_work_runs}

Several recent works have revisited the theory of bank runs in the wake of SVB's collapse (\cite{Jiang2023,Drechsler2023, Granja2024}). However, these works are concerned with self-fulfilling solvency runs connected purely to a positive franchise value. Our analysis is not tied to franchise value and we focus on how a given run is resolved, i.e., what step of Algorithm \ref{alg:clearing} it terminates at, which involves the study of partial runs and liquidity concerns. To highlight the overall similarities and clarify where our approach differs, we briefly discuss the model of \cite{Jiang2023} (closely related to \cite{Drechsler2023}) and then the model based on it in \cite{Granja2024}.

In \cite{Jiang2023}, a monopoly bank has assets with book value $1$, of which $c$ is cash, and $1-c$ are held in perpetuities with a unit value $p:=\frac{r_0}{r_f}$. In our terminology, the value of liquid assets is $x=c$ and the value of AfS assets is $sp=(1-c)p$. The bank has liabilities $L\leq 1$, split into insured and uninsured deposits, $L_I$ and $L_U$. Each unit deposit comes with a franchise value of $\delta := \frac{r_f - \mu(r_f)}{r_f} >0 $, where $\mu(r_f)<r_f$ is a lower rate paid to depositors, due to the bank's market power. A fraction $\pi$ of the uninsured depositors are `awake', and the premise of model is the following: if the bank's equity would be negative after a full withdrawal of $\pi L_U$, i.e., if $c+(1-c)p+(L-\pi L_U)\delta < L$, then all awake depositors run, and the bank is left insolvent (there are no withdrawals otherwise).

The above is equivalent to there being a run precisely if the leverage ratio $\lambda^{\delta}(0,0)$ exceeds a given threshold $\bar{\lambda}_{\mathrm{max}} := \frac{A^\delta(0,0)}{\pi L_U \delta}$, and, by definition, all awake depositors withdraw in that case. There is never a run if $\delta=0$, corresponding to $\bar{\lambda}_{\mathrm{max}}=\infty$. In our framework, there can still be runs with $\delta=0$, since we treat the threshold $\lm$ as a free variable to model depositors' tolerance (note that the above model also has a free variable, $\pi$, but with a different meaning). Moreover, we model how withdrawals erode the asset side through fire sale costs if AfS assets are sold as well as the additional costs and impact of remarking if the bank must tap into its HtM assets. The distinction between AfS and HtM assets in relation to run risk is also considered by  \cite{Granja2024}. However, rather than incorporating this within a model for the resolution of a run, as we do here, they study how an optimizing bank may use the HtM designation and a hedging instrument to deal with a capital requirement and the threat of a full insolvency run which is due to occur with a given probability.

In \cite{Granja2024}, the run itself is modelled as in \cite{Jiang2023} with all depositors awake ($L_U = L$, $\pi =1$) and the entire unit asset invested in the perpetuity ($c=0$). Let $p$ denote the price in a bad state where the bank's equity without franchise, $E=p-L$, is assumed negative, while $E^\delta = p - L + \delta L$ is positive (in the good state $E>0$, so nothing happens). The bad state occurs with probability $\mathbf{p}$, and, given this, the corresponding self-fulfilling run by all depositors occurs with probability $\varepsilon$. Now, the bank can designate a fraction $h$ of its assets as HtM to help pass a regulatory constraint $E_{\mathrm{reg}} := h + (1-h)p -L \geq \bar{e}$, where $\bar{e}$ is given. Moreover, at a given mark-up $\tau>1$, it can hedge away the run by paying $\tau (L-p)\mathbf{p}$ to receive $L-p$ in the bad state (and nothing in the good state): equity without franchise would then be non-negative, so the self-fulfilling run on franchise is ruled out.\footnote{Here it is implicitly assumed that the bank can pay for the hedge, e.g., by raising outside funds (despite the run threat), in a way that does not change the model mechanics.} With $\bar{e}:=1-L$ and $\theta$ small enough\footnote{Specifically, $\theta\le\theta^{*}:=\frac{(1-\varepsilon)\mathbf{p}\tau(L-p)}{1-(1-\varepsilon)h^{*}}$. We note that $\theta^*$ is smaller than the bound in \cite{Granja2024}, where there appears to be a minor discrepancy in the derivation.}, the highest expected payoff for the above inputs is obtained by (i) doing nothing if $E^\delta < e_1$, (ii) designating $h=1$ if $e_1 \leq E^\delta \leq e_2$, or (iii) doing the hedge and designating $h^*:=\frac{1-L}{1-p}$ if $E^\delta > e_2$, where the thresholds $e_1 \leq e_2$ are $e_1 = \frac{\theta}{(1-\varepsilon)\mathbf{p}}$ and $e_2=\frac{\tau}{\varepsilon}(L-p-\frac{(1-h^{*})\theta}{\tau \mathbf{p}})$.

These three cases are the central conclusion of the model in \cite{Granja2024}. Note that everything is determined by how large the (bad-state) equity $E^\delta$ is relative to the model parameters. In the first case, the bank accepts to fail with probability $\mathbf{p}$ due to the regulatory requirement. In the second case, it passes the requirement, but fails in the event of the run, occuring with probability $\mathbf{p}\varepsilon$. In the third case, it sidesteps the run. In practice, the bank may not be able to rule out an impending run as in case (iii), and it should not designate HtM in the way of case (ii) if the probability of the run is non-negligible. In Section \ref{sect:AfS-HtM}, we formulate a different optimization problem to study, for a given $\lm$, how much a bank can designate as HtM if it is to avoid remarking this after a shock to $p$, noting that the shock increases the leverage ratio and hence can spark a run in our framework. Since the HtM amount is part of our run mechanics, the bank's choice affects the nature of the run it is optimizing against.

\section{Illustration of the model mechanics}\label{sec:caseofsvb}

We start with a stylized example of how the model works. The outcomes of some representative simulations are depicted in Figure \ref{fig:demo_opt_bs}. For simplicity, herein we ignore any dead-weight losses from remarking HtM assets ($\alpha = 1$).

There are four groups of stacked bars, each representing the assets and liabilities of the bank under a specific scenario. We assume that $x=5$, $s=h=l=10$, $L_U=22$ and $L_I=5$, implying leverage of 4.375. The price impact function is linear with $b=0.01$. The first group of stacked bars refers to the optimal structure of the bank with high enough $\lambda_{max}=4.5$, that implies no run-off of the uninsured funding. For $\lambda_{max}=3.5$, i.e., with depositors considering a lower leverage as an acceptable level for the bank to be financially sound, depositors withdraw a fraction of funding that can be covered by cash ($x$) and about 1/4 of the available for sale assets. This will have a very limited impact on capital. However, in case depositors consider 2.5 as an acceptable leverage, a withdrawal of deposits would need to be covered also by a 30\% of the HtM portfolio. This would imply a remarking of the whole HtM with an additional impact on capital. As the last group of bars shows, the equilibrium withdrawal and, consequently, the devaluation of the HtM with capital impact, depends on the price impact function, which is steeper in this case ($f(\gamma)=1-0.015\gamma$).

\begin{figure}[H]
    \centering
    \includegraphics[width=0.6\textwidth]{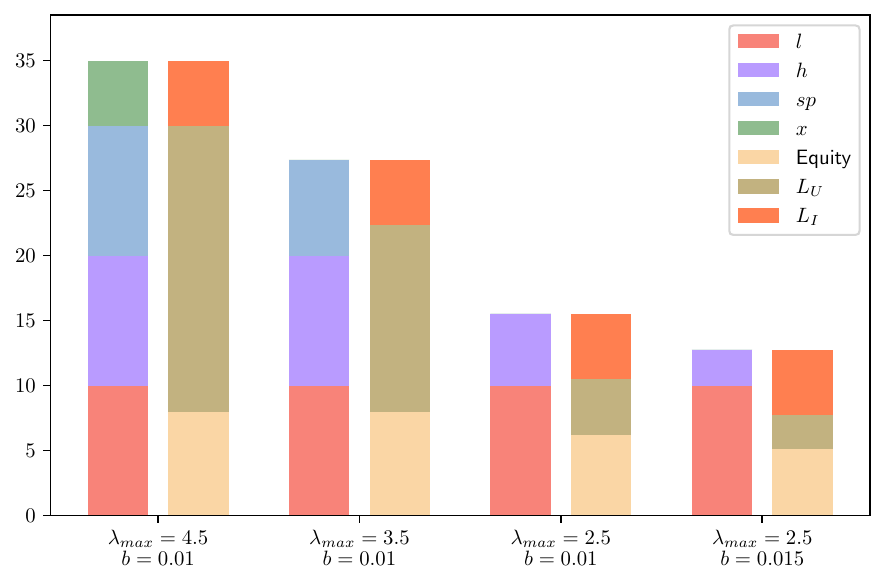}\vspace{-8pt}
    \caption{The figure shows the mechanics of the model based on a stylized bank balance sheet and different scenarios of $\lambda_{max}$ and price impact sensitivity $b$. }
    \label{fig:demo_opt_bs}
\end{figure}

Notably, the model can be calibrated using publicly available data to get realistic insights on banks' balance sheet vulnerabilities. To demonstrate how calibration can be implemented and how to draw conclusions about the banking system, we look at the case of Silicon Valley Bank (SVB) that experienced detrimental bank-run in March 2023.

The key trigger of the run was related to the unrealized losses in the securities portfolio of SVB. The losses --- related to changes in marked-to-market value of debt securities --- were excluded from SVB's capital but economically weighed on SVB's viability. Piles of unrealized losses may cross a tipping point whereby market participants change perception and expect HtM asset selling might be looming and amplify financial distress. To reflect on that, we assume that unrealized losses are deducted from capital and increase the leverage ratio.

The evolution of some key categories of SVB's balance sheet are shown in Table \ref{tab:svb_data}, based on the financial reports of the bank and the FRB report dissecting SVB case \citet{FRB2023}.  Between Q1 2020 and Q1 2022, i.e., one year before the collapse, total deposits grew more than threefold, from \$56 billion to \$181 billion. Only a small fraction of the funding base was insured deposits (\$9 billion out of the \$181 billion in Q1 2022). The absorbed funding was mostly invested into HtM securities (increase from \$10 billion to \$101 billion). When expectations about interest rate increases built, and eventually interest rates started to rapidly raise, the market value of the securities was gradually declining. However, thanks to the accounting treatment regarding how their value would be reflected in the financial results, this was only reflected in a build-up of the unrealized losses (increase from a gain of \$0.8 billion to a loss of \$15 billion in Q2 2022)
\begin{table}[!ht]
    \centering
    \resizebox{\textwidth}{!}{\begin{tabular}{ll|ccccccccp{2cm}p{2cm}|cp{2cm}}
    \hline
        ~ & ~ & \textit{In USD billion} & ~ & ~ & ~ & ~ & ~ & ~ & ~ & ~ & ~ & \textit{Ratio} & ~ \\ \hline
        ~ & ~ & Total deposits & Other funding & Insured deposits & Capital & Total assets & Cash & AfS & HtM & Unrealised Gains/Losses (HtM) & Unrealised Gains/Losses (AfS) & Tier 1 lev. ratio & Lev. ratio implied by Unrealised Gains/Losses \\ \hline\hline
        2020 & q1 & 56 & 8.9 & 5 & 10.1 & 75 & 8 & 20 & 10 & 0.8 & 1.6 & 7.4 & 6.0 \\ 
        ~ & q2 & 70 & 7.9 & 5 & 12.1 & 90 & 10 & 25 & 10 & 0.8 & 1.6 & 7.4 & 6.2 \\ 
        ~ & q3 & 80 & 6.5 & 5 & 13.5 & 100 & 12 & 28 & 12 & 0.8 & 1.6 & 7.4 & 6.3 \\ 
        ~ & q4 & 95 & 8.8 & 5 & 16.2 & 120 & 13 & 35 & 15 & 0.8 & 1.6 & 7.4 & 6.5 \\ \hline
        2021 & q1 & 110 & 11.7 & 5 & 18.3 & 140 & 16 & 30 & 40 & 0.0 & 0.0 & 7.6 & 7.6 \\ 
        ~ & q2 & 130 & 18.3 & 6 & 21.7 & 170 & 18 & 25 & 60 & 0.0 & 0.0 & 7.8 & 7.8 \\ 
        ~ & q3 & 152 & 10.0 & 7 & 23.0 & 185 & 21 & 25 & 80 & -0.5 & 0.0 & 8.0 & 8.2 \\ 
        ~ & q4 & 172 & 16.9 & 8 & 26.1 & 215 & 23 & 27 & 103 & -1.0 & 0.0 & 8.2 & 8.6 \\ \hline
        2022 & q1 & 181 & 17.3 & 9 & 26.7 & 225 & 22 & 27 & 101 & -7.5 & -1.5 & 8.4 & 12.7 \\ 
        ~ & q2 & 170 & 20.0 & 10 & 25.0 & 215 & 20 & 27 & 98 & -11.5 & -2.0 & 8.6 & 18.7 \\ 
        ~ & q3 & 162 & 28.5 & 10 & 24.5 & 215 & 19 & 27 & 95 & -16 & -3.0 & 8.8 & 39.2 \\ 
        ~ & q4 & 160 & 31.0 & 10 & 24.0 & 215 & 17 & 27 & 93 & -15 & -3.0 & 9.0 & 35.9 \\ \hline
    \end{tabular}}
    \caption{Balance sheet evolution of the SVB\newline Note: Numbers shown starting from the beginning of 2020 when the dynamics of assets and liabilities started to materially change. Unrealised Gains/Losses on neither HtM nor AfS portfolios are included into CET1, with the treatment of the AfS part following SVB's choice allowed by FED's enhanced prudential regulatory (EPR) framework; ``Lev. ratio implied by Unrealized Gains/Losses'' $=$ [Total assets]/([Capital]-[Unrealised Gains/Losses (HtM)]-[Unrealised Gains/Losses (AfS)]); ``Other funding'' $=$ calibrated such that balance sheet identity is preserved and leverage ratio reported by SVB ([Tier 1 ratio]) equals to the calculated leverage ratio (i.e., [Total assets]/[Capital]), ``AfS'' $=$ securities in available for sale accounting portfolios; ``HtM'' $=$ securities in held-to-maturity accounting portfolios. ``Other funding'' covers other types of liabilities, like own debt issued or negative market value of derivatives that are not subject to run risk and can be treated in the model as ``Insured deposits''.\newline Source: SVB financial reports and \citet{FRB2023}}\label{tab:svb_data}
\end{table}

We ran our framework for the 12 snapshots of SVB prior to its meltdown in 2023. Notably, for the price impact function, we assume $f(\gamma)=1-0.0005\gamma$ as a baseline case which means that \$10 billion of sold securities would have 50 bp impact on their market prices. However, since it is an ad-hoc calibration, we present the outcomes of simulations for a range of admissible values between 0.0001 and 0.002. Moreover, $\lambda_{max}$ is set to 7.5. Figure \ref{fig:svb_sim_alpha_unr_loss_withdrawal} shows that already at the the beginning of 2022 financial conditions of SVB became conducive to bankruptcy. In Q1, SVB would stay solvent but may already be considered illiquid and as of the subsequent periods, assuming a higher sensitivity of asset values in fire sales, the bank could be considered both illiquid and insolvent. The outcomes of the simulations indicate that, given the mounting unrealized losses, the financial conditions of SVB would deteriorate sharply.
\vspace{-6pt}
\begin{figure}[H]
    \centering
    \includegraphics[width=0.6\textwidth]{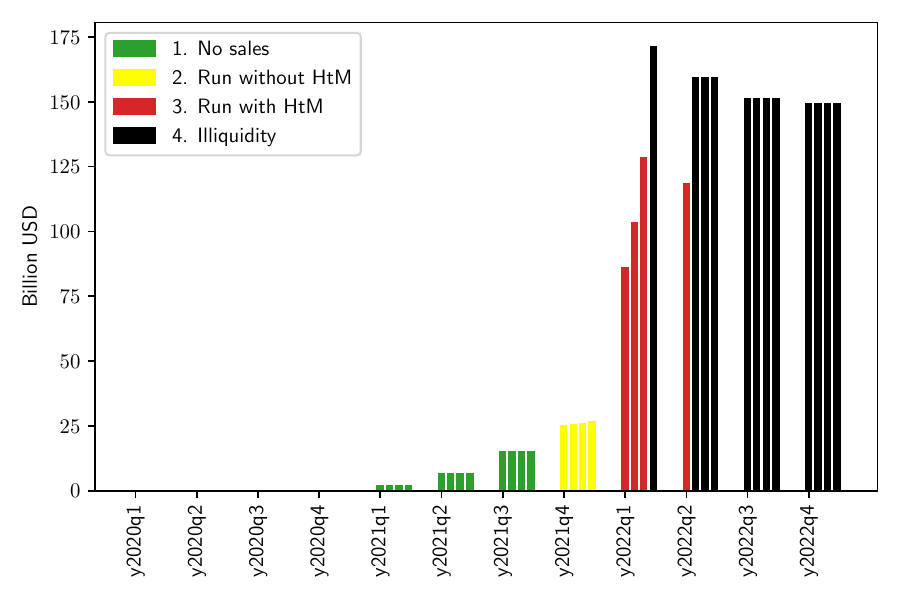}\vspace{-8pt}
    \caption{The figure shows equilibrium funding withdrawals from SVB (in USD billion) in a hypothetical scenario of unrealized losses in AfS and HtM portfolios being realized in the value of the securities portfolios and for balance sheets observed between Q1 2020, and Q4 2022 and for various parameterizations of price impact functions. For each period there is a group of bars, each of them corresponding to one parameter of the linear impact function ($b$) from the set $\{0.0001,\ 0.0002,\ 0.001,\ 0.002\}$. For instance, 0.0001 corresponds to 10 bp impact on asset prices when \$10 billion of securities are liquidated, like in \citet{Greenwood2015}. Max acceptable leverage ratio $=$ 7.5. Colored bars correspond to steps 1--4 of the Algorithm  \ref{alg:clearing}. Grey and black bars indicate insolvency differentiating liquidity and illiquidity state.}
    \label{fig:svb_sim_alpha_unr_loss_withdrawal}
\end{figure}

Notably, the qualitative conclusions from these simulations broadly hold in case some key parameters of the model are altered. Specifically, we ran a sensitivity analysis of the results to changes in the deadweight loss ($\alpha$) and maximum acceptable leverage ($\lambda_{max}$), whose exact values is less obvious, and is more challenging to derive from observables. The outcomes are shown in Figure \ref{fig:svb_sim_deadweight_unr_loss_withdrawal}. In a nutshell, the model consistently indicated Q1 2022 as a transition quarter between a viable balance sheet structure to a run-prone structure with imminent insolvency risk of the SVB. Some small differences are in the magnitude of the run in terms of the funding withdrawal and, in case of larger deadweight costs, the model signals an insolvency shifted by one quarter earlier.

\begin{figure}[H]
    \centering
    \includegraphics[width=0.9\textwidth]{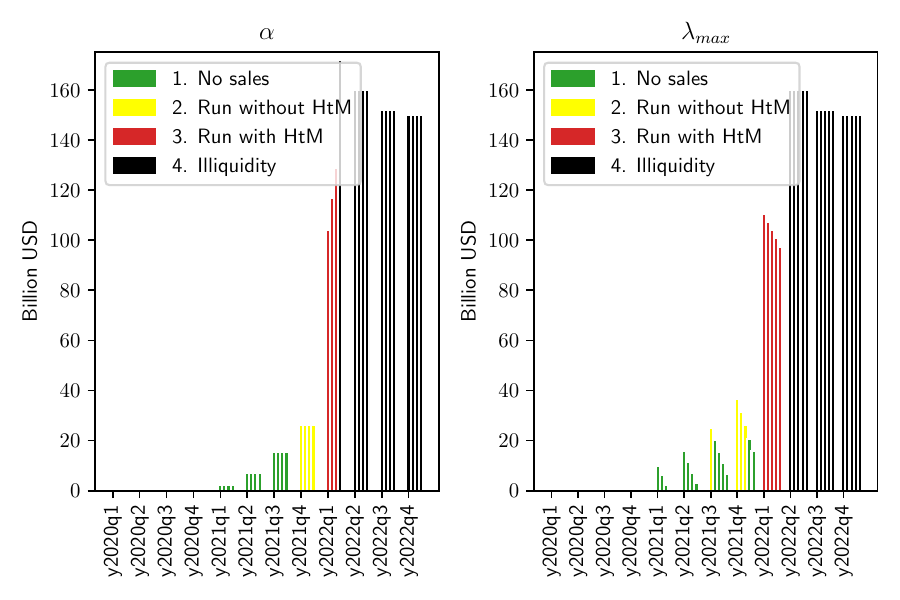}\vspace{-8pt}
    \caption{The figure shows sensitivity of equilibrium funding withdrawals from SVB (in USD billion) to deadweight loss ($\alpha$) and maximum acceptable leverage ($\lambda_{max}$) in a hypothetical scenario of unrealized losses in AfS and HtM portfolios being realized in the value of the securities portfolios and for balance sheets observed between Q1 2020, and Q4 2022 and for various parameterizations of price impact functions. For each period there is a group of bars, each of them corresponding to one parameter of $\alpha$ from the set $\{1.00,\ 0.98,\ 0.96,\ 0.94\}$ and maximum acceptable leverage $\lambda_{max}$ from the set $\{7.1,\ 7.3,\ 7.5,\ 7.7,\ 7.9\}$. The price impact function sensitivity is set to $0.0002$. Colored bars correspond to steps 1--4 of the Algorithm  \ref{alg:clearing}. Grey and black bars indicate insolvency differentiating liquidity and illiquidity state.}
\label{fig:svb_sim_deadweight_unr_loss_withdrawal}
\end{figure}

The above simulations report only the minimal clearing solution $(w^\downarrow,\gamma^\downarrow)$, which results in a minimum necessary volume of funding withdrawals that bring the bank back into region of acceptable leverage. However, as noted in Proposition~\ref{prop:exist} and Remark~\ref{rem:all_solutions}, the clearing problem may admit a lattice of solutions, bounded below by $(w^\downarrow,\gamma^\downarrow)$ and above by the maximal clearing solution $(w^\uparrow,\gamma^\uparrow)$. Whenever these two differ, it is \emph{possible} that depositors push the system into the worst-case equilibrium $(w^\uparrow,\gamma^\uparrow)$; this is in the spirit of the self-fulfilling runs of \cite{Diamond1983,GoldsteinPauzner2005}. To demonstrate the multiplicity of equilibria, Figure~\ref{fig:svb_minmax_quarters} contrasts the minimal and maximal solutions across the SVB dataset with $\lm = 7.5$, $b = 0.0005$, and $\alpha = 0.9$. Three regimes emerge in this data: through 2021Q1, the solution is unique and benign (no run arises); from 2021Q2 to 2021Q4, the two solutions diverge with little to no liquidations in the minimal clearing solution but with HtM remarking in the maximal clearing solution; from 2022Q1 onward, uniqueness of the illiquidity solution is recovered.
\begin{figure}[H]
\centering
\includegraphics[width=0.9\textwidth]{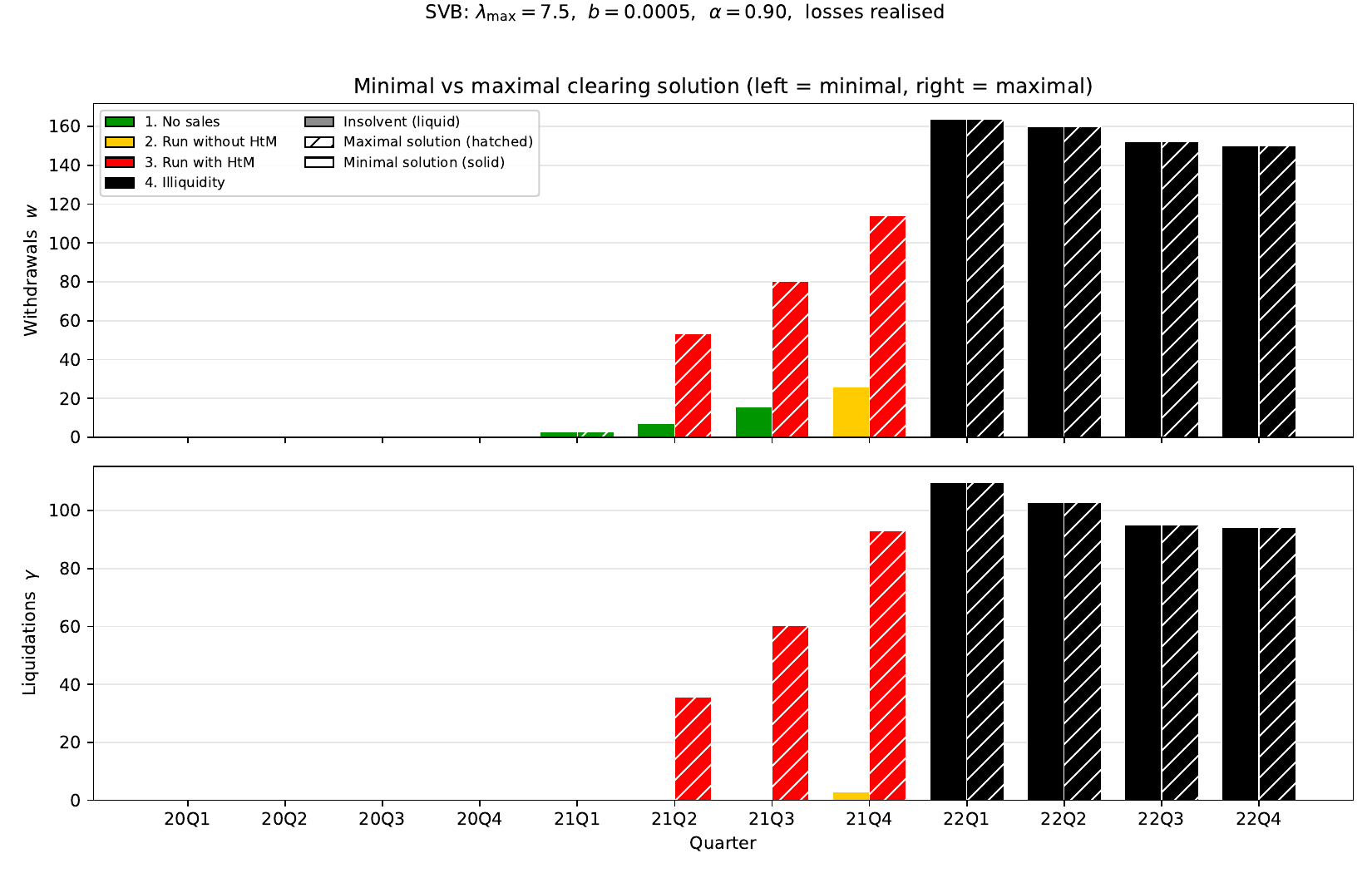}\vspace{-8pt}
\caption{Minimal (solid, left bar) versus maximal (hatched, right bar) clearing solution for SVB across the quarters of Table~\ref{tab:svb_data}, with $\lm = 7.5$, price-impact sensitivity $b = 0.0005$, dead-weight factor $\alpha = 0.9$, and unrealized losses realized. The top panel shows equilibrium withdrawals $w$, the bottom panel the liquidations $\gamma$. Bars are coloured by the step of Algorithm~\ref{alg:clearing} (green: no sales; amber: run without HtM remarking; red: run with HtM remarking; black: illiquidity), grey when insolvent but liquid. The gap between the two bars in 2021Q2--Q4 measures the run risk hidden by the minimal solution; it collapses into a unique illiquid solution from 2022Q1.}
\label{fig:svb_minmax_quarters}
\end{figure}

Figure~\ref{fig:svb_minmax_lambda} traces the minimal and maximal clearing solutions for 2021Q4 as functions of $\lm$ with the region of multiplicity of solutions shaded. When $\lm$ is low, the unique solution is a full run leading to illiquidity. As $\lm$ rises, the minimal solution improves through the steps of Algorithm~\ref{alg:clearing} while the maximal solution remains a HtM-remarking run. The shaded region ($\lm \in [6.6,11.3]$) is precisely the range in which a run is possible but not forced.
\begin{figure}[H]
\centering
\includegraphics[width=0.7\textwidth]{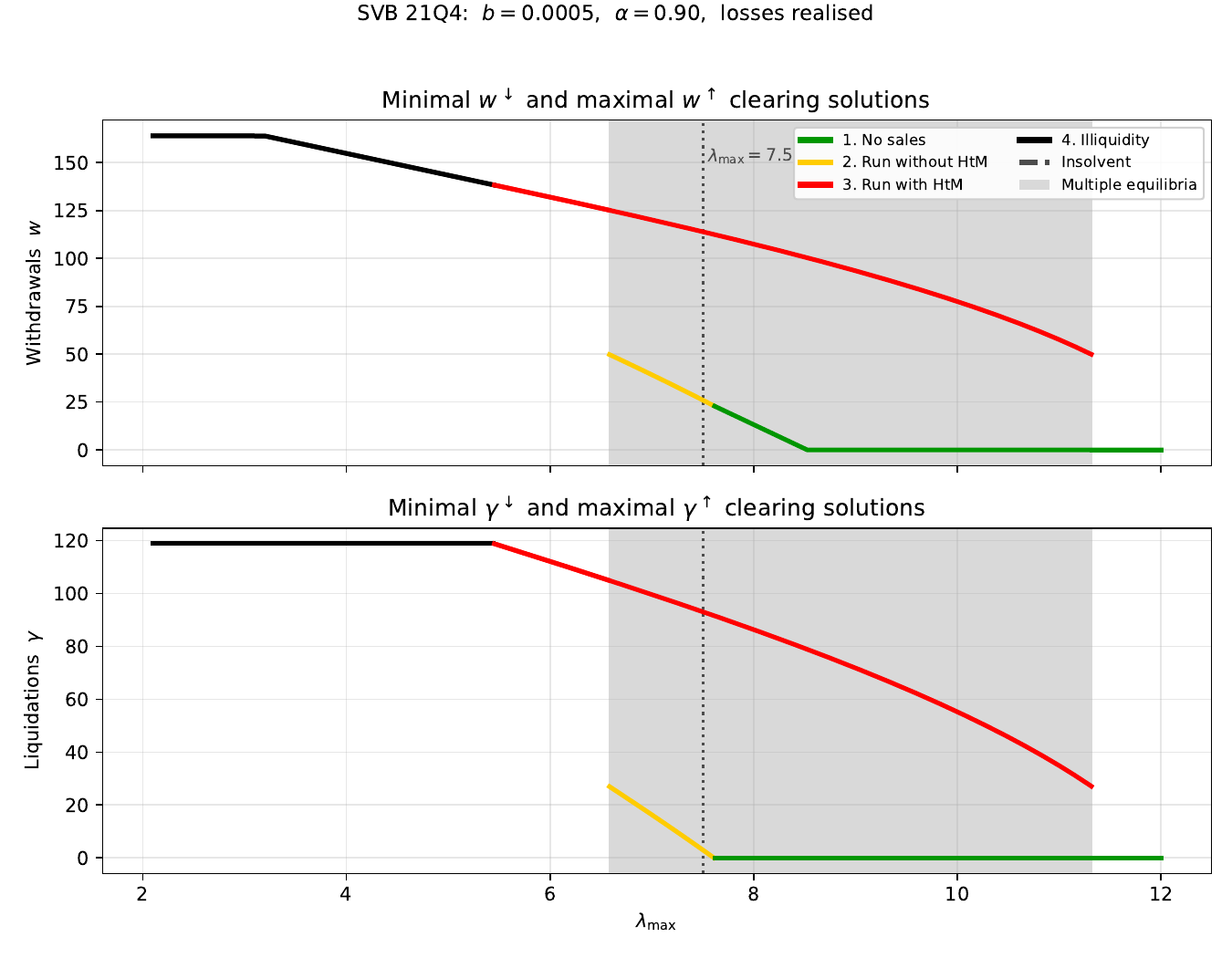}\vspace{-8pt}
\caption{Minimal $(w^\downarrow,\gamma^\downarrow)$ and maximal $(w^\uparrow,\gamma^\uparrow)$ clearing solutions for SVB in 2021Q4 as functions of the maximum acceptable leverage $\lm$ ($b = 0.0005$, $\alpha = 0.9$, unrealized losses realized). Curves are coloured by the step of Algorithm~\ref{alg:clearing} and dashed where the bank is insolvent. The shaded band marks the values of $\lm$ for which the minimal and maximal solutions differ (multiple equilibria); the dotted line marks the calibrated $\lm = 7.5$.}
\label{fig:svb_minmax_lambda}
\end{figure}

\section{On the balance between HtM and AfS}\label{sect:AfS-HtM}

Until now, we have taken for granted the bank's balance sheet composition. In this section, we endogenize the bank's decision to designate part of the marketable securities as HtM. Importantly, we keep all other aspects of the balance sheet fixed, so that it is only the allocation of the marketable securities between the AfS and HtM categories that the bank can vary.

There are simple and sound reasons for banks to rely on the HtM rules. Their business models (at least in the case of the more traditional universal banks) rely on maturity transformation, i.e., they invest in long-term projects financed by short-term funding. However, investment in long-term projects may be achieved via non-marketable loans or via bonds and equities, which are more liquid, frequently exchange-traded. At the same time, liquidity needs lead banks to also hold bonds to be able to quickly raise cash, and the holding period of those bonds may be short. Since these instruments are marked-to-market, they create volatility in the banks' profit and loss accounts. To decrease the variability of income from bonds and other securities, which are intentionally held for long-term investment purposes (i.e., held-to-maturity) regulators introduced accounting rules that allow banks to recognize these bonds at amortized costs. In the opposite direction, there are also natural pressures to not rely too heavily on HtM. First of all, it represents a real legal commitment to hold the assets to maturity and in turn involves some loss of flexibility. Moreover, whenever a bank looks to sell even a fraction of those assets, they would need to derecognize the whole HtM portfolio, thus forcing them to acknowledge unrealized losses while also signaling their inability to stay true to their commitment. 

\cite{Kim2023} provide empirical evidence that banks' use of the HtM category often appears aimed at optimizing around capital requirements and accounting measures such as net income and economic value of equity, especially when there are concerns about negative valuation impact on solvency (e.g.~the rising rates in 2022--23). This section presents a stylized version of such a setting, where a bank looks to maximize its HtM allocation in anticipation of a shock. Without accounting for the fact that this could ignite a run, the bank's incentive is to hold as much as possible in the HtM category, in order to insulate the bank from price fluctuations up until maturity and makes the balance sheet look as strong as possible. When incorporating our bank run model, however, it becomes necessary for the bank to consider its ability to honor the commitment of holding HtM securities to maturity. Specifically, the bank should have enough liquid assets to cover potential funding withdrawals without the need to remark and liquidate HtM portfolios in most plausible stress scenarios. Otherwise, the possible benefits from the HtM accounting rules, making income less sensitive to revaluation shocks, are questionable.

\subsection{Maximising the HtM designation subject to price shocks}

Consider a bank with a balance sheet as in Figure \ref{fig:balance_sheet}. The bank has total assets $A$ of which $\bar A := A - x - \ell$ are held in marketable securities that may be designated as some combination of AfS and HtM. To account for the impact of a possible future devaluation of the marketable securities, we introduce a simple one-period model, wherein the HtM and AfS allocations are decided at time $0$, before a price shock which arrives at time $1$. Given the new price, a potential run is resolved within the equilibrium formulation from Section \ref{sec:clearing}. As is implicit in the word shock, we stress that we shall only consider downside risk.

\vspace{2pt}\begin{figure}[H]
\centering
 \includegraphics[width=0.55\textwidth]{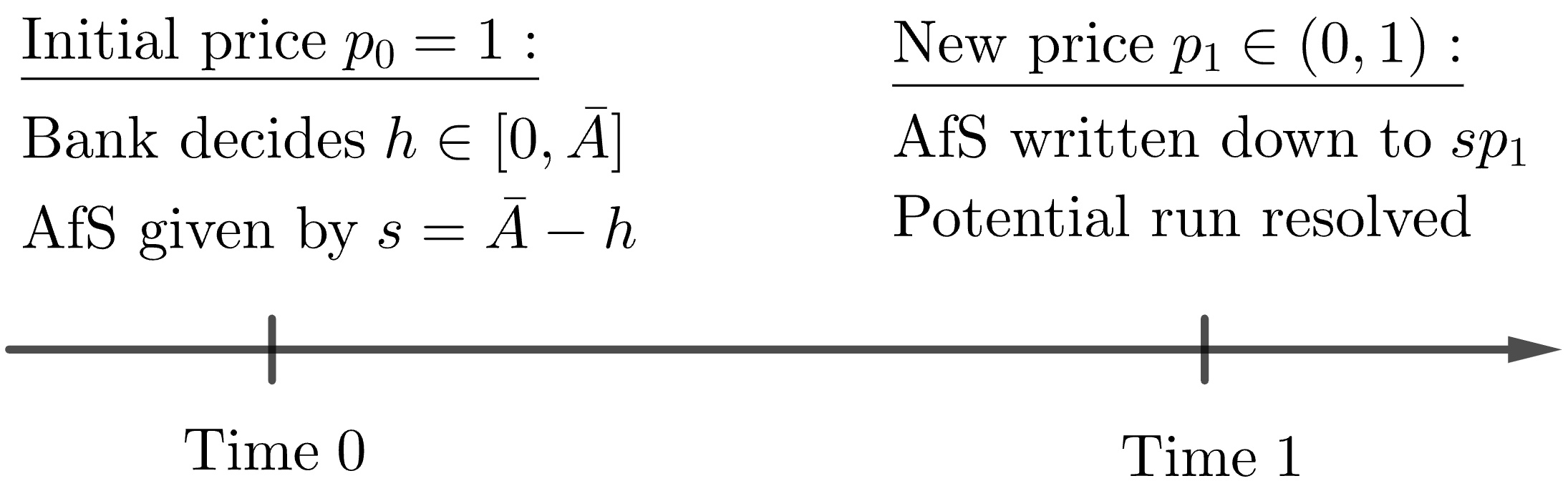}
\caption{Graphical depiction of the one-period model. Given the price shock, the bank suffers realized losses on the AfS assets, while the value of the HtM assets is unchanged. The unrealized losses can only be an issue if the state of the bank is such that a run forces it to remark the HtM assets.}
\label{fig:price_shock}
\end{figure}

As soon as a potential run could necessitate the sale of marketable securities, care must be taken if the bank aims to act responsibly when deciding on its HtM designation. First of all, regulations call for the bank's intent and ability to hold the HtM assets to maturity. This implies that the bank should seek to avoid having to remark the HtM assets for a reasonable range of price shocks. Secondly, since the entire holdings must be remarked at once, remarking the HtM assets can cause a significant shock to the balance sheet from both unrealized losses and dead-weight losses (as well as fire sale losses). Thirdly, such remarking may be a bad signal to the outside world and may remove the flexibility to recognize those same assets as HtM in the future. Thus, the bank management has incentives to avoid remarking of HtM assets for what it considers to be likely values of the price shock.

The above considerations lead us to formulate the following optimization problem based on the one-period model depicted in Figure \ref{fig:price_shock}: the bank maximizes the amount of HtM assets that it holds, at time $0$, subject to having enough AfS to cover liquidity needs in a potential run at time $1$, for any price shock within some given range. That is, the bank solves for the maximum $h^*\in[0,\bar{A}]$ at time $0$ such that there is no need to remark HtM assets at time $1$ if the new price $p_1$ is above or equal to a given threshold price. From here on, we let $p_1$ stand for this threshold price, as it is the only value of the price at time $1$ that we need to consider (larger prices amount to less equilibrium sales). The optimization problem can then be written as
\begin{equation}\label{eq:optim-h}
 h^* = \max\{h \in [0,\bar A] \; | \; \gamma^\downarrow(p_1,\lm) \leq \bar A - h\}
\end{equation}
or, equivalently, $h^*= \bar A - s^*$ where
\begin{equation}\label{eq:optim-s} s^*=\min\{s \in [0,\bar A] \; | \; \gamma^\downarrow(p_1, \lm) \leq s\}.
\end{equation}
Here we have made explicit the dependence of the (minimal) equilibrium liquidations $\gamma^\downarrow$ on the threshold price $p_1$ and the maximum accepted leverage ratio $\lm$. If the bank cannot satisfy the constraint, then everything is held as available-for-sale, so we set $h^*=0$ and $s^*=\bar{A}$ in that case.

\begin{remark}
Considered as an internal risk management problem, the threshold price may be seen as a reflection of the banks' risk tolerance, expressed in terms of accepted negative valuation limits. This could, for instance, be related to Value-at-Risk or Expected Shortfall limits.
\end{remark}

For tractability, we will model the price impact by a linear inverse demand function. Given the values $p_1$ and $\lm$, the next result fully characterizes the optimal behavior of the bank.

\begin{proposition}[Maximal HtM designation]\label{prop:optim} Assume $\lm > 2$
and let $\bar{l} := (\lm-1)/ \lm$. For a given threshold price $p_1 \in (0,1)$, at time $1$, we assume a linear inverse demand function, which takes the form $f(\gamma) = p_1(1 - b\gamma)$ with $b < 1/[(\lm-1)\bar A]$. We then have the following two cases.
\begin{description}
    \item[Case 1] If $ L_U \leq x$ or $L \leq x + \bar{l}(\bar A + \ell)$, then, by designating everything as HtM, the bank can insulate itself from the price shock and have no risk of a run on the marketable securities, so the bank will hold $h^*:=\bar{A}$ and $s^*:=0$.

    \item[Case 2] If $ L_U > x$ and $L > x + \bar{l}(\bar A + \ell)$, then the bank chooses to hold $h^*:=\bar{A} - s^*$ as HtM, where $s^* = s^*(p_1,\lm)$ is specified by \eqref{eq:min_AfS_explicit} below.
    \end{description}
\end{proposition}

In Case 2 of Proposition~\ref{prop:optim}, the optimal value $s^*(p_1,\lm)$ is such that the bank needs to consider the worst case of two potential fictitious runs when faced with a price shock $p_1$: (i) a partial withdrawal run so that not all liabilities are withdrawn resulting in $s_{\mathrm{PW}}(p_1,\lm)$ and (ii) a full withdrawal run resulting in $s_{\mathrm{FW}}(p_1,\lm)$. If neither case can be satisfied by the available assets ($\bar{A}$) then the bank must hold all assets as AfS. This amounts to \begin{equation}\label{eq:min_AfS_explicit}
s^*(p_1,\lm) = \min \{ s_{\mathrm{PW}}(p_1,\lm) , s_{\mathrm{FW}}(p_1,\lm) , \bar{A} \}.
\end{equation}
This logic follows similarly to the proof of Proposition~\ref{prop:clearing_algo}.
The amount of AfS assets $s_{\mathrm{PW}}$ or $s_{\mathrm{FW}}$, with partial or full equilibrium withdrawals, is defined explicitly as follows. Firstly, assuming $w^* = L_U$ so that all assets are withdrawn in the run, we get:
\begin{align*}
s_{\mathrm{FW}}(p_1,\lm) &:= \begin{cases} \max \{s_1 , s_2 \} \; , \;  &\text{if }\; \max \{s_1 , s_2 \} \leq \bar A \text{ and }  p_1 > 2b(L_U-x) \\
    + \infty &\text{else} \end{cases}
\end{align*}
where 
\begin{equation*}
s_1 := \frac{1-\sqrt{1-2b(L_U-x)/p_1}}{b} , \quad s_2 := \frac{\bar A+\ell - p_1 s_1 (1-b s_1)-L_I/\bar{l}}{1-p_1(1-bs_1)}.
\end{equation*}
Secondly, assuming $w^* < L_U$ for a partial withdrawal, we arrive at:
\begin{align*}
s_{\mathrm{PW}}(p_1,\lm) & := \begin{cases} \displaystyle \frac{p_1 - \bar{l} - M}{b p_1} &\text{if }\; \displaystyle \frac{p_1 - \bar{l} - M}{b p_1} \leq  \bar s  \text{ and } p_1 \geq \bar{l} + b[L-x-\bar{l}(\bar A+\ell)] + C \\
    + \infty &\text{else} \end{cases} 
\end{align*}
with
\begin{equation*}
\quad M =\sqrt{(p_1 - \bar{l})^2 - 2p_1 b [L-x-\bar{l}(\bar A+\ell)]},\quad 
C = \sqrt{ b\bigl(L-x-\bar{l}[\bar{A}+\ell]\bigr)\bigl(2\bar{l} + b(L-x-\bar{l}[\bar{A}+\ell])\bigr)  },
\end{equation*}
where $\bar s:=0$ if $L_U \leq G(0)$ and $\bar s:=\bar{A}$ if $L_U \geq G(\bar{A})$ while $\bar{s}$ is the (unique) solution to $ G(s)=L_U$ on $(0,\bar{A})$ if $G(0)< L_U < G(\bar{A})$, for\vspace{-2pt}
\begin{align*}
&\qquad \qquad G(s) := \lm L - (\lm-1)[x + \bar\gamma(s)\bar{f}(\bar\gamma(s))+(s-\bar\gamma(s))f(\bar\gamma(s))+\bar A + \ell - s]\quad \text{and} \\
& \bar\gamma(s) := \frac{p_1[(\lm-1)bs-1]+\sqrt{p^2_1[(\lm-1)bs-1]^2 + 4 \lm^2 p_1 b (\bar{l}-\frac{1}{2})( L-x-\bar{l}[\bar A + \ell-s(1-p_1)])}}{p_1 b(\lm-2) }.
\end{align*}

We wish to note that the dead-weight loss $\alpha$ does not enter into these expressions. This is because the AfS holdings $s^*$ are exactly such that the HtM portfolio is not remarked given the price shock $p_1$.

\subsection{Case studies of the optimal HtM}\label{app:beyondSVB}

To illustrate how we can analyze the HtM holdings of a given bank within the setting of Proposition \ref{prop:optim}, we consider two case studies. First, we refer to the stylized example of a bank balance sheet composition introduced at the beginning of Section \ref{sec:caseofsvb}. Second, we use the case of SVB to show what level of HtM our model would imply to ensure SVB's resilience to funding shocks related to the run-off risk. As in Section~\ref{sec:caseofsvb}, herein we ignore any dead-weight losses from remarking HtM assets (i.e., $\alpha = 1$).

In the first experiment, we determined the optimal level of HtM for the four scenarios considered in section 4. The first 3 scenarios are differentiated by the maximum acceptable leverage set at 4.5 (above the observed leverage in the example), 3.5 (below the observed level and with p=1.0 implying only tapping liquidity from AfS) and 2.5 (necessitating to mobilize the HtM to cover the funding withdrawal). In those 3 scenarios, the sensitivity of prices to transacted volumes of securities is set to $b=0.01$. The fourth scenario, also with acceptable leverage at 2.5 is characterized by an additionally stressed market condition captured by a more sensitive price impact function, i.e., with $b=0.015$. The results are shown in Figure \ref{fig:styl_opt_htm_range_lambdas}.

Clearly, for the two cases of the baseline $\lambda_{max}=4.5$ and $\lambda_{max}=3.5$, the optimal HtM is above the status quo level of 10. In the 3.5 case, only some very severe shocks to the initial asset prices, i.e., for $p<0.8$, would mean that to cover the funding outflows the stylized bank would need to keep all securities as available for sale. The two more stressful cases, in which the depositors accept leverage at 2.5 and not higher, the bank is bound to fully designate securities as available for sale.

\begin{figure}[H]
    \centering
    \includegraphics[width=0.65\textwidth]{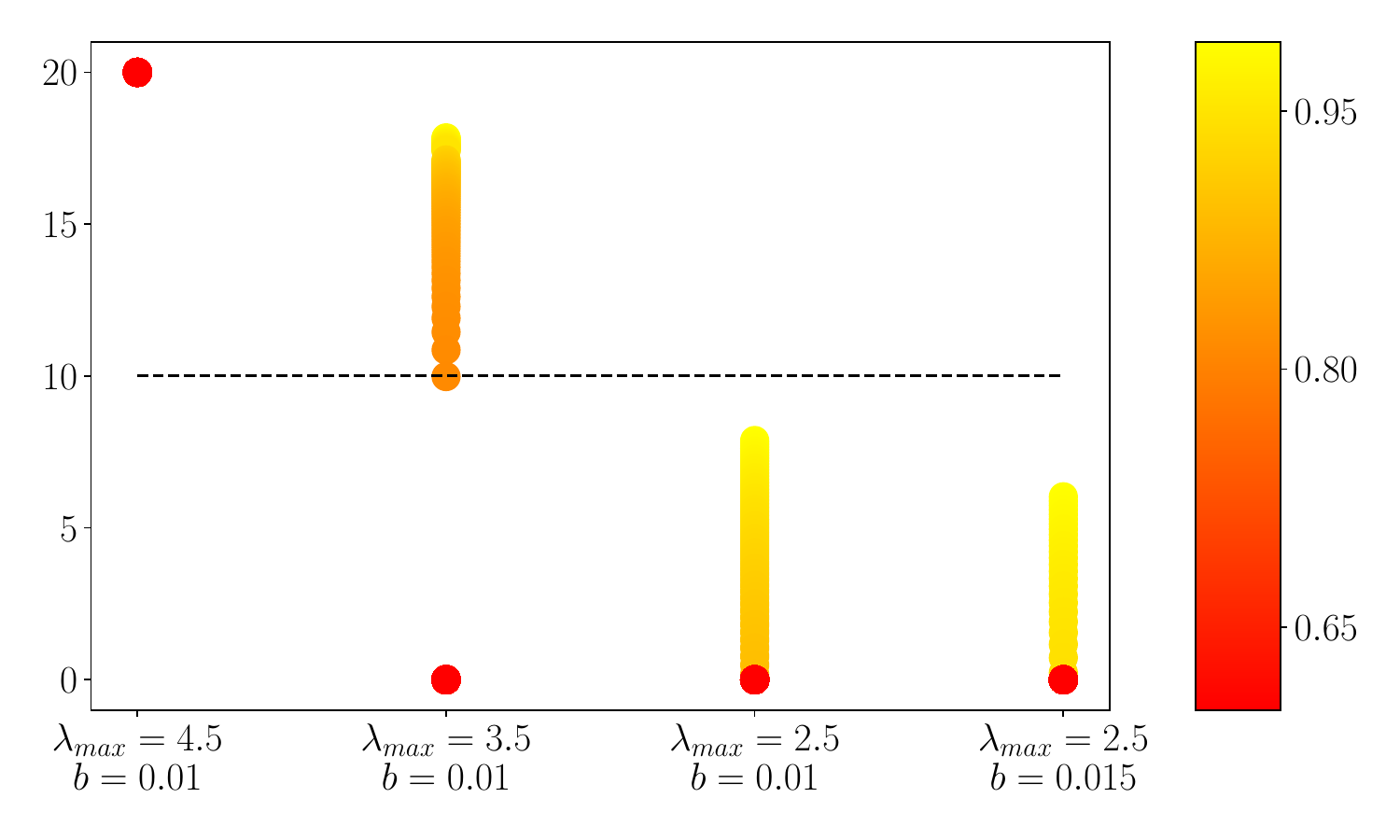}
    \caption{The figure shows theoretically optimal HtM portfolios of the bank in the stylized example. Each colored dot represents an optimal HtM (y-axis) for a scenario of $\lambda_{max}$ and price impact sensitivity $b$ (x-axis). Dashed black line represents the level of the HtM portfolios reported by the bank.}
    \label{fig:styl_opt_htm_range_lambdas}
\end{figure}

We ran the second experiment based on the case study of SVB with an objective to see how sensitive the optimal level of HtM is to the target leverage ratio around the actual leverage ratios of SVB prior to the March 2023 events. In this way, we can see how much the depositors' tolerance to banks' leverage would need to change, so that the reported holdings of the banks would no longer be commensurate with the overall balance sheet structure or consistent only with some very large price shocks ($p$) assumed in the Proposition \ref{prop:optim}. 

Specifically, we calculated the optimal HtM for a range of leverage ratios around those reported by the banks and for two snapshots of the data (end of 2021 and end of 2022). The outcomes are shown in Figure \ref{fig:svb2_opt_htm_range_lambdas}.
There are two qualitatively different regions in this figure. One corresponds to the leverage ratios implying that the optimal level of HtM is not below the reported volumes of HtM. It corresponds to red dots lying above the black dotted line and indicates that the HtM holding does not violate the ability of the bank to lock securities in portfolios where it is less straightforward to tap liquidity from. The remaining leverage ratios constitute a region where the optimal HtM is not consistent with the reported HtM for some initial shocks to the value of securities portfolios. This means that the bank may not have the ability to hold such a large volume of securities in the HtM portfolio, depending on the economic outlook assumed by the bank in its asset and liability management. Moreover, as SVB was approaching the default in March 2023, the observed leverage ratio was deeper in the region where admissible asset price shocks would mean that the bank allocated an excessive amount of securities into the HtM portfolios. Graphically, some dots in Fig. \ref{fig:svb2_opt_htm_range_lambdas}, bottom pane, corresponding to $\lambda_{max}=7.9$ and $\lambda_{max}=8.0$, lie below the dashed black line indicating the volume of the HtM as of Q4 2022. To summarize, the optimization introduced in Proposition \ref{prop:optim} is a practical tool to detect some inconsistencies in banks' balance sheet structures that may expose them to bank-run risk.

\begin{figure}[H]
    \centering
    \includegraphics[width=0.65\textwidth]{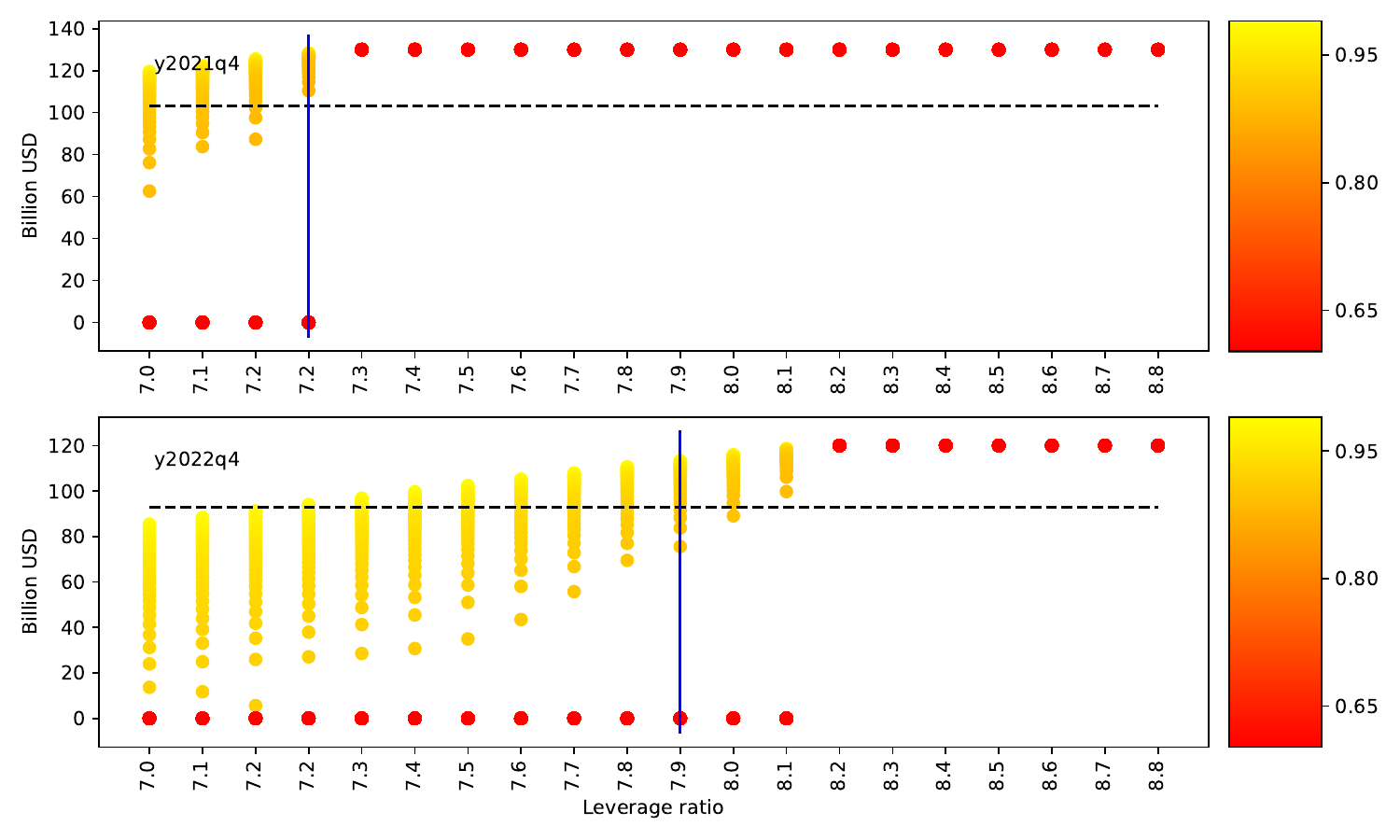}
    \caption{The figure shows theoretically optimal HtM portfolios of SVB for maximum accepted leverage ratios from a range in x-axis. Solid blue line represents the reported Tier 1 leverage ratio. Dashed black line represents the level of the HtM portfolios reported by the bank.}
    \label{fig:svb2_opt_htm_range_lambdas}
\end{figure}

\section{Discussion}\label{sec:concl}

In this work, we have built a model that can help explain the mechanics of bank runs in relation to structural vulnerabilities of the bank balance sheet, i.e., a large share of uninsured liabilities, accounting rules that may initially `hide' revaluation risk of bank assets, and insufficient liquidity buffers to cover funding withdrawals, especially in the presence of fire-sale risk. The parsimonious setup of the model makes it straightforward to devise simple indicators of banks' exposure to run risks, as quantified by equilibrium withdrawals of funding that are easily computed via our explicit algorithm.

By inputting publicly available balance sheet data, the model allows us to analyze the build-up of run risk at SVB in the years before its collapse. First, we demonstrate that, during 2022, its balance sheet composition and growth were creating conditions for an imminent bank run. Whether this applies already at the beginning of 2022 or later in the year, and whether it comes with illiquidity or insolvency implications, depends on investors' perception about unrealized losses materializing or not and on the volume of cash that could be raised by liquidating securities portfolios, i.e., on the sensitivity of market prices to the volume of off-loaded securities. Second, we show that  SVB's choice to park majority of deposits in HtM asset portfolios violated prudent risk management principles, as the bank lost its ability to mobilize liquidity necessary to prevent a rational run by depositors. Finally, we are able to monitor when and how liquidity problems transform from manageable, meaning that they are avoidable by investing new deposit volumes into AfS assets, to a state where default is imminent. 

In closely related recent works, the empirical findings of \cite{Granja2023} and \cite{Kim2023} have cast doubt on whether---or at least to what extent---banks take seriously the intent and ability to hold HtM assets until they mature. Our conclusions reveal that, under reasonable assumptions, such negligence was indeed foreboding of SVB's failure when one factors in a simple model of depositor runs---with the implied outcomes being detrimental already throughout 2022. Additional case studies for banks in other jurisdictions could provide further insights.

The severity of SVB's collapse and the ensuing crisis have spurred a serious public debate about the future of the HtM framework and possible regulatory responses. Not least, there have even been calls to abandon it, but \cite{Kim2023} note that this may not be feasible in view of enduring support for the original motivations behind the rules, particularly in relation to banks' economic hedge of interest rate risk through its deposit franchise in scenarios that do not lead to a run. Similarly, \cite{Granja2023} stresses the need to carefully consider trade-offs, noting in particular the aforementioned hedge and the concern that mark-to-market valuations may have been a key propagator of contagion during the Global Financial Crisis. In view of this, \cite{Kim2023} mention increased enforcement of the existing GAAP restrictions on the intent and ability to hold HtM designations and, in parallel, \cite{Granja2023} suggests the need for more thorough scrutiny and evaluation of the reasonableness of banks' claims about their ability to hold assets to maturity. This traces out an important role for regulators.

There exist supervisory tools to assess banks' total balance sheet sensitivity to shocks, for instance the EBA technical standards to \emph{`evaluate if there is a large decline in the net interest income or in the economic value of equity that could trigger supervisory measures'}.\footnote{See https://www.eba.europa.eu/activities/single-rulebook/regulatory-activities/supervisory-review-and-evaluation-process-srep-1} However, the above considerations raise the issue of how to analyze and assess the reasonableness of banks' HtM portfolios. We have taken a tractable first step in that direction, by endogenizing a given banks' HtM designation subject to our model of depositor runs. This allowed us to derive a simple conceptual framework for the reasonableness of their choice, or desire, to hold certain levels, in a way that is easy to assess quantitatively. Specifically, for a given balance sheet, we characterize the bank's maximal HtM designation such that it is able to safely hold on to these assets in a potential depositor run driven by a negative price shock of some specific size. For a bank's HtM designation to be reasonable, it should be in line with a plausible range of price shocks. By calibrating to the balance sheets of SVB over time, we obtain clear warning signs about the reasonableness of its HtM designations already from Q3 2021 or Q1 2022, depending on what we assume about the depositors' propensity to run. Notably, these indications appear before the size of unrealized losses at SVB ballooned. Thus, our model highlights potential concerns about SVB's reasonable usage of the HtM framework before this could be assessed simply from the size of their HtM holdings and corresponding unrealized losses.

\appendix

\section{Proofs of main results}\label{app:proofs}

\subsection{Proof of Proposition \ref{prop:exist}}

One readily confirms that the two mappings \eqref{eq:map_w} and \eqref{eq:map_gamma} are non-decreasing in $(w,\gamma) \in [0,L_U]\times [0,s+\alpha h]$. As the domain of $\Phi$ defined by \eqref{eq:map_w}--\eqref{eq:map_gamma} is a complete lattice, the claim therefore follows from Tarski's fixed point theorem.

\subsection{Proof of Proposition \ref{prop:clearing_algo}}

By Proposition \ref{prop:exist} there exists a minimal clearing solution $(w^*,\gamma^*) \in [0,L_U]\times[0,s+\alpha h]$, provided the bank is solvent. The left-hand side of the solvency condition \eqref{eq:solvent} reads as
\[
x + \gamma \bar f(\gamma) + (s-\gamma)f(\gamma) + h + \ell,\quad \text{for} \quad \gamma\in[0,s],\quad \text{and}
\]
\[
x + \gamma \bar f(\gamma)  + (s+\alpha h-\gamma) f(\gamma) + \ell, \quad \text{for}\quad \gamma\in(s,s+\alpha h].
\]
Following Remark \ref{rem:increas_non-increase}, these are both non-increasing functions of $\gamma$ on the respective domains. Moreover, at $\gamma=s$, there is a jump of size $(\alpha f(s)-1)h \leq 0$, since $f(s)\leq 1$ by Assumption \ref{ass:idf} and $\alpha \in [0,1]$. Consequently, if the bank was already insolvent at some level of liquidations $\gamma$, it is also insolvent for all larger values. It therefore suffices to check for solvency at the termination of the algorithm, since the algorithm is increasing in the value of $\gamma^*$.

By construction, we must have that either $\gamma^*=0$ (no sales), $\gamma^*\in (0,s]$ (run without re-marking of HtM), $\gamma^*\in (s,s+\alpha h)$ (run with re-marking of HtM), or $\gamma^*=s+\alpha h$ (illiquidity). Studying these case-by-case, proceeding through seven steps in increasing order with respect to the values of $(\gamma^*,w^*)$, we will be able to conclude that the minimal clearing solution is indeed realized by Algorithm \ref{alg:clearing}.

\textbf{Step 1: No sales.} Assume $\gamma^* = 0$. Then $w^* = \Phi_w(0) = L_U \wedge [\lm L - (\lm-1)(x+sp+h+\ell)]^+$. This is a clearing solution if and only if $w^* \leq x$. This, in turn, holds if and only if $L_U \leq x$ or $\lm L - (\lm-1)(x+sp+h+\ell) \leq x$. The latter holds if only if
\[
L \leq \frac{x}{\lm} + \frac{\lm-1}{\lm}(x+sp+h+\ell)=x + \frac{\lm-1}{\lm}[sp+h+\ell].
\]

\textbf{Step 2: Run without re-marking HtM, case (i).} Suppose $\gamma^* \in(0,s]$. Then $w^*\in (x,L_U]$ with $L_U>x$. For this step, assume $w^*\in(x,L_U)$. Since $\gamma^* \in(0,s]$, we can see that $w^*=\Phi_w(\gamma^*)$ holds if and only if
\begin{equation}\label{eq:HtM_2*}
L_U \geq \lm L - (\lm-1)(x + \gamma^* \bar{f}(\gamma^*) + (s-\gamma^*)f(\gamma^*) + h + \ell).
\end{equation}
Note that $w^*$ equals the right-hand side of \eqref{eq:HtM_2*}. Moreover, $\gamma^*$ must satisfy $\gamma^* \bar{f}(\gamma^*) = w^* - x$ and it is the unique such solution, since the left-hand side is strictly increasing in $\gamma^*$ on $[0,s+h]$ (by Assumption \ref{ass:idf}). Inserting $w^*=x+\gamma^* \bar{f}(\gamma^*)$ in \eqref{eq:HtM_2*} and recalling that the right-hand side equals $w^*$, we obtain
\[
w^* = L - (1-\frac{1}{\lm})((s-\gamma^*)f(\gamma^*)+h+\ell).
\]
Thus, the liquidation $\gamma^*\in(0,s]$ satisfies
\begin{equation}\label{eq:HtM_2_b}
\gamma^* \bar f(\gamma^*) + (1-\frac{1}{\lm})(s-\gamma^*)f(\gamma^*) = L - x - (1-\frac{1}{\lm})(h+\ell),
\end{equation}
and it must be the unique solution to this equation on $(0,s]$, since the left-hand side is strictly increasing in $\gamma^*$ on $[0,s]$ by Assumption \ref{ass:idf-2} and Remark \ref{rem:increas_non-increase}. This is possible if and only if 
\begin{equation}\label{eq:HtM_2_c}
L-x-(1-\frac{1}{\lm})(h+\ell) \in \bigl[(1-\frac{1}{\lm})sp , s\bar f(s)\bigr].
\end{equation}
Consequently, we have a clearing solution if and only if both \eqref{eq:HtM_2_c} and \eqref{eq:HtM_2*} hold with $\gamma^*$ in \eqref{eq:HtM_2*} being the unique solution to \eqref{eq:HtM_2_b}.

    \textbf{Step 3: Run without re-marking HtM, case (ii).} Now assume $\gamma^*\in(0,s]$ and $w^*=L_U$. Then $\gamma^*$ satisfies $\gamma^* \bar f(\gamma^*) = w^* - x= L_U - x$. As the left-hand side is strictly increasing in $\gamma^*$ on $[0,s+h]$, we have a unique solution which is in $(0,s]$ if and only if $L_U\in (x,x+s\bar f(s)]$. With $\gamma^*\leq s$, we have $w^*=\Phi_w(\gamma^*)=L_U$ if and only if $L_U \leq \lm L - (\lm-1)(x+\gamma^*\bar f(\gamma^*) + (s-\gamma^*)f(\gamma^*)+h+\ell)$. Writing $L=L_I+L_U$, this re-arranges to
\begin{equation}\label{eq:HtM_1}
L_I \geq (1-\frac{1}{\lm})[(s-\gamma^*)f(\gamma^*)+h+\ell].
\end{equation}
Consequently, $(\gamma^*,w^*)$ is a clearing solution if and only if $L_U\in(x,x+s\bar{f}(s)]$ and \eqref{eq:HtM_1} holds for the unique solution $\gamma^* \in(0,s]$ of $\gamma^* \bar f(\gamma^*)= L_U - x$.

 \textbf{Step 4: Re-marking HtM, case (i).} Suppose $\gamma^*\in (s,s+\alpha h)$. Then $w^* \in (x,L_U]$. For this step we assume $w^*\in (x,L_U)$. We have $w^*=\Phi_w^*(\gamma^*)$ if and only if
 \begin{equation}\label{eq:HtM_1_a}
 L_U \geq \lm L - (\lm-1)(x+\gamma^*\bar f(\gamma^*) + (s+\alpha h-\gamma^*)f(\gamma^*) + \ell),
 \end{equation}
and $w^*$ is then given by the right-hand side of \eqref{eq:HtM_1_a}. Noting that $\gamma^*$ must be the unique solution of $\gamma^* \bar f(\gamma^*)=w^*-x$ (where the left-hand side is strictly increasing in $\gamma^*$ on $[0,s+h]$), we can insert this in \eqref{eq:HtM_1_a} and solve for
\[
w^*=L - (1-\frac{1}{\lm})((s+\alpha h-\gamma^*)f(\gamma^*)+\ell).
\]
In turn, $\gamma^*\in(s,s+\alpha h)$ must solve
 \begin{equation}\label{eq:HtM_1_b}
 \gamma^* \bar f(\gamma^*) + (1-\frac{1}{\lm})(s+\alpha h-\gamma^*)f(\gamma^*) = L-x-(1-\frac{1}{\lm})\ell,
 \end{equation}
and it must the unique such solution since the left-hand side is strictly increasing on $[0,s+\alpha h]$ by Assumption \ref{ass:idf-2}. This is feasible if and only if
 \begin{equation}\label{eq:HtM_1_c}
 L-x-(1-\frac{1}{\lm})\ell \in \bigl[s\bar f(s) + (1-\frac{1}{\lm})\alpha hf(s) , (s+\alpha h)\bar f(s+\alpha h)\bigr].
 \end{equation}
In conclusion, we have a clearing solution if and only if \eqref{eq:HtM_1_c} and \eqref{eq:HtM_1_a} hold, when $\gamma^*$ in \eqref{eq:HtM_1_a} is given by the unique solution to \eqref{eq:HtM_1_b}.

\textbf{Step 5: Re-marking HtM, case (ii).} Now assume $\gamma^*\in (s,s+\alpha h)$ and $w^*=L_U$. Then $\gamma^* \bar f(\gamma^*)=L_U-x$, which is possible if and only if $L_U\in(x,x+(s+\alpha h)\bar f(s+\alpha h))$. Moreover, we see that $\Phi_w(\gamma^*)=L_U$ holds if and only if
\begin{equation}\label{eq:HtM_2}
L_I \geq (1-\frac{1}{\lm})[(s+\alpha h-\gamma^*)f(\gamma^*)+\ell].
\end{equation}
We thus have a clearing solution if and only if $L_U\in(x,x+(s+\alpha h)\bar f(s+\alpha h))$ and the unique solution $\gamma^*$ to $\gamma^* \bar f(\gamma^*)=L_U-x$ satisfies \eqref{eq:HtM_2}.

\textbf{Step 6: Illiquidity.} Finally, assume $\gamma^* = s + \alpha h$. Then 
\[
w^* = \Phi_w(s+\alpha h) = L_U \wedge [\lm L - (\lm - 1) (x + (s+\alpha h)\bar f(s+\alpha h) + \ell)].
\]
This is a clearing solution if and only if $w^* - x \geq (s+\alpha h)\bar f(s+\alpha h)$. Given the expression for $w^*$, this holds if and only if either
    \begin{align*}
    &\lm L - (\lm - 1)(x + (s+\alpha h)\bar f(s+\alpha h) + \ell) \geq L_U \quad \text{and} \quad L_U \geq x+ (s+\alpha h)\bar f(s+\alpha h),\quad \text{or}\\
   &\lm L - (\lm - 1)(x + (s+\alpha h)\bar f(s+\alpha h) + \ell) < L_U
   \quad \text{and} \quad L \geq x + (s+\alpha h)\bar f(s+\alpha h) + (1 - \frac{1}{\lm})\ell,
    \end{align*}
since the last inequality is equivalent to
\[
\lm L - (\lm - 1)\bigl(x + (s+\alpha h)\bar f(s+\alpha h) + \ell)\bigr) \geq x + (s+\alpha h)\bar{f}(s+\alpha h).
\]
This completes the proof.

\subsection{Proof of Propositions \ref{prop:franchise_1} and \ref{prop:franchise_2}}

\begin{proof}[Proof of Proposition \ref{prop:franchise_1}]
Note that we have
\begin{equation}\label{eq:not_restore_lm}
	A^{\delta}(w,\gamma)-\lm\,E^{\delta}(w,\gamma)=(1-\kappa)(L-w)-(\lm-1)E(0,\gamma),
\end{equation}
and hence
\begin{equation}\label{eq:kappa_geq_leq}
	\lambda^{\delta}(0,0)>\lm \iff (1-\kappa)L>(\lm-1)E(0,0).
\end{equation}
Thus, with $\kappa \geq 1$ and  $E(0,0)\geq 0$ in \eqref{eq:kappa_geq_leq}, we must have $\lambda^\delta(0,0)\leq \lm$, confirming that a run is never initiated. Next, if $E(0,0)< 0$ and $\lambda^\delta(0,0)>\lm$, then a run is initiated, and, since $\kappa \geq 1$, we have that the right-hand side of \eqref{eq:not_restore_lm} is increasing in $w$. Together with $E(0,\gamma)\leq E(0,0)$, this implies that a leverage ratio of $\lm$ cannot be restored, so we get full withdrawal requests $w^*=L_U$.
\end{proof}
\begin{proof}[Proof of Proposition \ref{prop:franchise_2}]
	The franchise value $\delta L$ is non-marketable, so it cannot be liquidated to cover withdrawals. Hence the mapping $\Phi_\gamma$ for the quantity sold in \eqref{eq:map_w} is unchanged. Regarding the map \eqref{eq:map_gamma} that gives the volume of withdrawals to restore $\lm$, for any given quantity sold $\gamma^*$, notice first that $\lambda^{\delta}(w,\gamma^*)=\lm$ is equivalent to
	\begin{align*}
		L-w&=(\lm-1)\big(A(0,\gamma^*)-L\big)+(\lm-1)\delta\,(L-w) \\
		&=(\lm-1)\big(A(0,\gamma^*)-L\big)+\kappa\,(L-w).
	\end{align*}
This re-arranges to $L-w=(\hat\lambda_{\max}-1)(A(0,\gamma^*)-L)$ 
	and, in turn,
	\[
	 w= L -(\hat\lambda_{\max}-1)(A(0,\gamma^*)-L)=\hat\lambda_{\max}L-(\hat\lambda_{\max}-1)A(0,\gamma^*),
	 \]
	 so the map \eqref{eq:map_gamma} with franchise value $\delta L$ becomes
		\[
	\Phi_w^{\delta}(\gamma^*)=L_U\we\big[\hat\lambda_{\max}L-(\hat\lambda_{\max}-1)A(0,\gamma^*)\big]^{+}.
	\]

In the case $E(0,0)\leq 0$, since $\kappa <1$, it follows from \eqref{eq:kappa_geq_leq} that we must have $\gamma^\delta(0,0)>\lm$, so a run is automatic. Moreover, we have $\hat\lambda_{\max}L-(\hat\lambda_{\max}-1)A(0,\gamma^*)\geq L$, as $\hat{\lambda}_{\mathrm{max}}>1$, so $\Phi_w^{\delta}(\gamma^*)=L_U$ for all $\gamma^*\geq0$, which confirms the first claim that there is a full run with $w^*=L_U$.

In the case $E(0,0) > 0$, both $\lambda^\delta(0,0)$ and $\lambda(0,0)$ are well-defined, and one readily sees from the definitions that, with $\kappa<1$, we have $\lambda^{\delta}(0,0)>\lm$ if and only if $\lambda(0,0)>\hat\lambda_{\max}$. Hence, a run being initiated is equivalent to a run being initiated for the same bank with zero franchise value and leverage threshold $\hat\lambda_{\max}$. By the above, the clearing maps $\Phi^\delta_w$ and $\Phi^\delta_\gamma$ correspond exactly to the latter situation, so the second claim follows. Finally, note that the baseline solvency condition \eqref{eq:solvent} can be written in shorthand as $A(0,\gamma^*)>L$, where $w^*$ is the equilibrium withdrawal requests. With the franchise value, this condition changes to $A(0,\gamma^*)+\delta(L-w^*)>L$.
\end{proof}

\subsection{Proof of Proposition \ref{prop:optim}}

Suppose first that $ L_U \leq x$ or $L \leq x + \bar{l}(\bar A + \ell)$. In the first situation, there is no risk of a run on the marketable securities. If, instead, we are in the other situation, then a run is possible depending on the HtM versus AfS designation. However, by taking $s^*=0$, we have $\lm L - (\lm-1)(x+s^*p_1 + h + \ell) =  \lm L - (\lm-1)(x+\bar{A}+ \ell) \leq  x$ at time $1$, no matter what $p_1$ is. With this choice, Proposition \ref{prop:clearing_algo} therefore gives that the clearing solution is of the `No sales' type, for any $p_1\in(0,1)$, and hence $s^*=0$ is the minimizer of \eqref{eq:optim-s}, as claimed.

From here on, suppose instead that $ L_U > x$ and $L > x + \bar{l}(\bar A + \ell)$. Since $ sp_1 + h\leq \bar{A} $, for any choice of $s\in [0,\bar{A}]$, it follows that, at time $1$, we have $\lm L - (\lm-1)(x+sp_1 + h + \ell) < x$ (along with $L_U>x$), so, by Proposition \ref{prop:clearing_algo}, we cannot have a `No sales' (minimal) clearing solution. Thus, we can proceed by identifying the feasible regions yielding that the (minimal) clearing solution belongs to either of the two `Run without re-marking HtM' scenarios in Proposition \ref{prop:clearing_algo}. We refer to these regions as partial or full withdrawals, and we denote the minimal attainable AfS over each by, respectively, $s_\mathrm{PW}$ or $s_\mathrm{FW}$ (assigning the value $+\infty$ if the region is empty).

We begin by characterizing the partial withdrawal region. Since $p_1 \leq 1$, and since we assume $L > x + \bar{l}(\bar A + \ell)$, by Proposition \ref{prop:clearing_algo} partial withdrawals are feasible for $s\in[0,\bar A]$ if and only if
   \begin{align}
	&L_U \geq \lm L - (\lm-1)\bigl(x+	\bar\gamma(s)  \bar f(	\bar\gamma(s) ) + (s-\bar \gamma(s))f(\bar \gamma(s))+\bar A-s+\ell \bigr), \label{eq:LU_PW}\\
 	&L-x-\bar{l}(\bar A-s+\ell) \leq s\bar f(s),\quad \text{and} \label{eq:L_bound_PW}\\
	&
	\bar\gamma(s) \bar f(	\bar\gamma(s) ) + \bar{l} (s-	\bar\gamma(s) )f(\bar\gamma(s)) = L - x - \bar{l} (\bar A-s+\ell), \label{eq:gamma_bar_PW}
\end{align}
for some $	\bar\gamma(s) \in(0,s]$. Using $\bar f(s)  = p_1 (1-bs/2) $ in \eqref{eq:L_bound_PW}, we obtain the quadratic expression
\begin{equation}
	-\frac{b}{2}p_1 s^2 + (p_1 -\bar l)s \geq L-x-\bar{l}(\bar A + \ell).
\end{equation}
As the right-hand side is strictly positive, we can confirm that this holds for $s\geq 0$ if and only if
\begin{equation}\label{eq:PW_constraint1}
	\frac{p_1 - \bar l - M}{bp_1} \leq s \leq \frac{p_1 - \bar l + M}{bp_1} \quad \text{with}\quad M:=\sqrt{(p_1 - \bar{l})^2 - 2bp_1 \bigl(L-x-\bar{l}(\bar A+\ell)\bigr)}
\end{equation}
and
\[
p_1 > \bar{l} + \sqrt{2bp_1\bigl(L-x-\bar{l}(\bar A + \ell)\bigr)}.
\]
The latter can be seen to hold if and only if
\begin{equation}\label{eq:PW_constraint2}
p_1 \geq \bar l + b\bigl(L-x-\bar{l}(\bar A + \ell)\bigr) + C
\end{equation}
for $C$ as in the statement of the proposition.
Now, for any $s$ in the above range (recalling also that $\bar l s p_1 \leq L-x-\bar{l}(\bar A-s+\ell)$), by Step 2 in the proof of Proposition \ref{prop:clearing_algo}, there is a unique $\bar\gamma(s) \in(0,s]$ satisfying \eqref{eq:gamma_bar_PW}. Inserting the expressions for $f$ and $\bar{f}$ in \eqref{eq:gamma_bar_PW}, we obtain a quadratic equation
\[
\bar{\gamma}^2 p_1 (\bar{l}-\frac{1}{2})b + \bar\gamma p_1 [1 - \bar{l}(sb+1)] = L-x - \bar{l}[\bar A + \ell - s(1-p_1)],
\]
in  $\bar\gamma$. Knowing that each $s$ in the above range corresponds to a unique $\bar{\gamma}\in [0,s]$, we can solve the above equation to yield the expression for $s\mapsto \bar \gamma(s)$ in the statement of the proposition. Further, we can observe that this map is strictly increasing in $s$, as follows directly from $f(\bar \gamma(s)) <1 $ and the fact that, for fixed $s$, the left-hand side of \eqref{eq:gamma_bar_PW} is strictly increasing in the value of $\bar\gamma(s)$ on $[0,s]$. This in turn ensures that the right-hand side of \eqref{eq:LU_PW} is strictly increasing in $s$. Let $G(s)$ denote the right-hand side of \eqref{eq:LU_PW}. If $G(\bar{A}) \leq L_U$, then all $s\in[0,\bar{A}]$ satisfy \eqref{eq:LU_PW}, and we set $\bar{s}:=\bar{A}$. If $G(0) \geq L_U$, then at most $s=0$ can satisfy \eqref{eq:LU_PW}, and we set $\bar{s}:= 0$. If $G(0) < L_U < G(\bar{A}) $, then we can define $\bar{s}\in (0,\bar{A})$ to be the unique value for which there is equality in \eqref{eq:LU_PW}. With these definitions, we have that $\eqref{eq:LU_PW} $ holds for a given $s\in(0,\bar{A}]$ if and only if $s \leq \bar{s}$. Noting that \eqref{eq:PW_constraint1}--\eqref{eq:PW_constraint2} enforces $s>0$, we conclude that the feasible region for partial withdrawals is given by the constraints \eqref{eq:PW_constraint1}--\eqref{eq:PW_constraint2} and $s\leq \bar{s}$. If this region is non-empty (in particular implying $\bar s \in(0,\bar{A}]$), then clearly the minimal $s$ over the region is $s=(p_1 - \bar{l}-M)/bp_1$.

Next, we turn to full withdrawals. Given $s\in [0,\bar{A}]$, let $\gamma^*\in[0,\bar{A}]$ denote the corresponding quantity sold (in the minimal clearing solution). Since $L_U>x$, by Proposition \ref{prop:clearing_algo} full withdrawals are feasible if and only if $\gamma^* \bar f(\gamma^*) = L_U - x$ with
   \begin{equation}\label{eq:FW_constraints}
    s\bar f(s) \geq L_U -x \quad \text{and}\quad L_I \geq \bar{l}\bigl( (s-\gamma^*(s))f(\gamma^*(s)) + \bar A -s  + \ell\bigr)
   \end{equation}
Writing out $\bar{f}(\gamma)=p_1(1-b\gamma/2)$, the above equality yields a quadratic equation
\[
-p_1 \frac{b}{2}(\gamma^*)^2 + p_1 \gamma^* = L_U -x
\]
in $\gamma^*$. Since $L_U >x$, this has positive solutions $\gamma^*>0$
if and only if $p_1 > 2b(L_U-x)$.
We cannot have $\gamma^*\geq 1/b$, as our assumptions enforce $b < 1/\bar A$, so we would have $\gamma^*>\bar{A}$. Solving for $\gamma^*$ thus gives the result that the only value allowing for $\gamma^*\in [0,\bar{A}]$ is
\begin{equation}\label{eq:FW_gamma}
\gamma^*=\frac{1 -\sqrt{1 -  2b(L_U-x)/p_1}}{b}\quad \text{with}\quad   p_1 > 2b(L_U-x).
\end{equation}
It furthermore follows from these observations that
$s\bar f(s) \geq L_U -x$ holds for a given $s\in[0,\bar{A}]$ if and only if $s\geq \gamma^*$ with $\gamma^*$ given by \eqref{eq:FW_gamma}. Since $\bar{f}(\gamma^*)<1$, the second constraint in \eqref{eq:FW_constraints} holds if and only if
\begin{equation}\label{eq:FW_last_constraint}
s\geq \frac{\bar A+\ell - \gamma^* p_1  (1-b \gamma^*)-L_I/\bar{l}}{1-p_1 + p_1 b\gamma^*},
\end{equation}
where we have written out the expressions for $f$ and $\bar f$. Set $s_1:=\gamma^*$ and let $s_2$ denote the right-hand side of \eqref{eq:FW_last_constraint}. For full withdrawals to be feasible, we need $p_1 > 2b(L_U-x)$ and $\max \{s_1,s_2\} \leq \bar{A}$. In that case, the feasible values are $s\in[0,\bar{A}]$ with $s \geq \max\{s_1, s_2\}$, so the minimum of \eqref{eq:optim-s} is achieved at $\max\{s_1, s_2\}$, as desired. This completes the proof.

\bibliographystyle{agsm}
\bibliography{dbs_references}

\end{document}